\documentclass[%
reprint,
superscriptaddress,
amsmath,amssymb,
aps,
amsthm,
pra,
]{revtex4-1}

\usepackage{graphicx}% Include figure files
\usepackage{dcolumn}% Align table columns on decimal point
\usepackage{bm}% bold math
\usepackage{braket}%
\usepackage{theorem}
\usepackage{multirow}
\usepackage{enumerate}
\usepackage{framed}
\usepackage{bbm}
\usepackage{algorithm}
\usepackage{algorithmic}

\usepackage[dvipdfmx,colorlinks,citecolor=blue,linkcolor=blue,urlcolor=blue,bookmarksopenlevel=4]{hyperref}

\usepackage{color}

\renewcommand{\algorithmicrequire}{\textbf{Input:}}
\renewcommand{\algorithmicensure}{\textbf{Output:}}
\theorembodyfont{\upshape}

\newtheorem{thm}{Theorem}
\newtheorem{proposition}[thm]{Proposition}
\newtheorem{remark}[thm]{Remark}

\newtheorem{lemma}[thm]{Lemma}
\newtheorem{cor}[thm]{Corollary}

\newtheorem{proof}{Proof}

\newcommand{\hA}{\hat{A}}
\newcommand{\hB}{\hat{B}}
\newcommand{\hC}{\hat{C}}

\newcommand{\hE}{\hat{E}}

\newcommand{\hG}{\hat{G}}
\newcommand{\hH}{\hat{H}}

\newcommand{\hP}{\hat{P}}

\newcommand{\hT}{\hat{T}}

\newcommand{\hX}{\hat{X}}
\newcommand{\hY}{\hat{Y}}
\newcommand{\hZ}{\hat{Z}}
\newcommand{\ha}{\hat{a}}

\newcommand{\he}{\hat{e}}

\newcommand{\hx}{\hat{x}}
\newcommand{\hy}{\hat{y}}

\newcommand{\hPi}{\hat{\Pi}}
\newcommand{\hPhi}{\hat{\Phi}}

\newcommand{\hrho}{\hat{\rho}}

\newcommand{\hsigma}{\hat{\sigma}}

\newcommand{\mA}{\mathcal{A}}

\newcommand{\mH}{\mathcal{H}}
\newcommand{\mI}{\mathcal{I}}

\newcommand{\mM}{\mathcal{M}}

\newcommand{\mR}{\mathcal{R}}
\newcommand{\mS}{\mathcal{S}}

\newcommand{\mZ}{\mathcal{Z}}

\newcommand{\PC}{P_{\rm C}}
\newcommand{\PE}{P_{\rm E}}
\newcommand{\PI}{P_{\rm I}}

\newcommand{\ident}{\hat{1}}

\newcommand{\Real}{\mathbf{R}}

\newcommand{\POVM}{\mM}

\newcommand{\QED}{\hspace*{0pt}\hfill $\blacksquare$}

\newcommand{\argmin}{\mathop{\rm argmin}}
\newcommand{\Tr}{{\rm Tr}}
\newcommand{\rank}{{\rm rank}}
\newcommand{\supp}{{\rm supp}}
\newcommand{\Ker}{{\rm Ker}}

\def\gauss_sym#1{{\lfloor #1 \rfloor}}

\newcommand{\opt}{{\star}}
\newcommand{\sopt}{{\circ}}
\newcommand{\DPme}{$\rm DP_{me}$}
\newcommand{\PCopt}{Q}
\newcommand{\PCU}[1]{\overline{\PCopt}^{\rm #1}}
\newcommand{\PCUp}{\overline{\PCopt}}
\newcommand{\PCUip}{\overline{\PCopt_p}}
\newcommand{\ul}[1]{\underline{#1}}

\newcommand{\PCLp}{\ul{\PCopt}}
\newcommand{\PCLip}{\ul{\PCopt_p}}

\newcommand{\Proja}{\ul{P_+}}
\newcommand{\Projb}{\overline{P_+}}
\newcommand{\tp}{\tilde{p}}
\renewcommand{\dag}{\dagger}
\renewcommand{\L}{{\rm L}}
\newcommand{\R}{{\rm R}}
\newcommand{\SRM}{{\rm SRM}}
\renewcommand{\dag}{\dagger}
\newcommand{\half}{{\frac{1}{2}}}
\newcommand{\inv}{-}

\begin{document}

\preprint{APS/123-QED}

\title{Upper and Lower Bounds on Optimal Success Probability of Quantum State Discrimination
 with and without Inconclusive Results}%

\affiliation{%
 Quantum Information Science Research Center, Quantum ICT Research Institute, Tamagawa University,
 Machida, Tokyo 194-8610, Japan
}%
\affiliation{%
 School of Information Science and Technology,
 Aichi Prefectural University,
 Nagakute, Aichi 480-1198, Japan
}%

\author{Kenji Nakahira}
\affiliation{%
 Quantum Information Science Research Center, Quantum ICT Research Institute, Tamagawa University,
 Machida, Tokyo 194-8610, Japan
}%

\author{Tsuyoshi \surname{Sasaki Usuda}}
\affiliation{%
 School of Information Science and Technology,
 Aichi Prefectural University,
 Nagakute, Aichi 480-1198, Japan
}%
\affiliation{%
 Quantum Information Science Research Center, Quantum ICT Research Institute, Tamagawa University,
 Machida, Tokyo 194-8610, Japan
}%

\author{Kentaro Kato}
\affiliation{%
 Quantum Information Science Research Center, Quantum ICT Research Institute, Tamagawa University,
 Machida, Tokyo 194-8610, Japan
}%

\date{\today}% It is always \today, today,

\begin{abstract}
 We propose upper and lower bounds on the maximum success probability for
 discriminating given quantum states.
 The proposed upper bound is obtained from
 a suboptimal solution to the dual problem of
 the corresponding optimal state discrimination problem.
 We also give a necessary and sufficient condition for the upper bound
 to achieve the maximum success probability;
 the proposed lower bound can be obtained from this condition.
 It is derived that a slightly modified version of the proposed upper bound
 is tighter than that proposed by
 Qiu {\it et al.} {[Phys. Rev. A \textbf{81}, 042329 (2010)]}.
 Moreover, we propose upper and lower bounds on the maximum success probability
 with a fixed rate of inconclusive results.
 The performance of the proposed bounds are evaluated through numerical experiments.
\end{abstract}

% PACS 03.67.Hk: Quantum communication
\maketitle

%\tableofcontents

\section{Introduction}

%% Background
Discrimination of quantum states is
a basic and important problem in the field of quantum information theory.
The objective of this work is to distinguish between a given finite set of known quantum states
as well as possible.
As is well known, no measurement can discriminate perfectly between non-orthogonal states;
thus, the problem is to find a measurement that minimizes or maximizes
a certain optimality criterion.
Since the pioneering work of Helstrom, Holevo, and Yuen {\it et al.}
\cite{Hol-1973,Hel-1976,Yue-Ken-Lax-1975},
quantum state discrimination problems with several criteria
have been widely investigated.

%% minimum-error measurement, optimal inconclusive measurement
The success probability is one of the most used criteria for discriminating quantum states.
A quantum measurement maximizing the success probability,
which is called a minimum-error measurement,
has been widely investigated.
%has been investigated, and closed-form analytical expressions have been obtained
%for some cases (see, e.g., \cite{Eld-2003-unamb,Her-2007,Kle-Kam-Bru-2010,Ber-Fut-Fel-2012}).
However, closed-form analytical expressions for minimum-error measurements
have only been obtained in some particular cases
(e.g., \cite{Bel-1975,Ban-Kur-Mom-Hir-1997,Usu-Tak-Hat-Hir-1999,Bar-2001,
And-Bar-Gil-Hun-2002,Cho-Hsu-2003,Eld-For-2001}).
Another criterion is based on the inconclusive probability;
a quantum measurement maximizing the success probability with a fixed failure
(i.e., inconclusive) probability
which is called an optimal inconclusive measurement,
has also been investigated \cite{Che-Bar-1998-inc,Eld-2003-inc,Fiu-Jez-2003}.
A minimum-error measurement and an unambiguous measurement that maximizes the success probability
can be regarded as special cases of optimal inconclusive measurements.
Obtaining an optimal inconclusive measurement is generally a more difficult task
than obtaining a minimum-error measurement.
In fact, closed-form analytical expressions for optimal inconclusive measurements
are only known for very special cases
(e.g., \cite{Her-2012,Nak-Usu-Kat-2012-GUInc,Bag-Mun-Oli-Ber-2012,Nak-Kat-Usu-2015-inc,Her-2015-inc}).
Instead of analytical approaches, we can use numerical methods.
It is known that the design of an optimal success probabilities
can be treated as a positive semidefinite programming problems \cite{Eld-Meg-Ver-2003}.
In many cases, an optimal value can be computed in polynomial time
by well known algorithms for solving semidefinite programs
such with interior point methods.
However, in large scale problems, these methods require the vast amount of calculation.

%% Upper / lower bounds
Instead of computing an exact optimal success probabilities,
several previous studies have given its upper and/or lower bounds
\cite{Hay-Leu-Smi-2005,Mon-2007,Hay-Kaw-Kob-2008,Mon-2008,Qiu-2008,Tan-Erk-Gio-Guh-2008,Tys-2009-2,
Qiu-Li-2010,Sug-Has-Hor-Hay-2009}.
These methods are especially useful for large scale problems
of which it is hard to compute an exact value within feasible time;
for example, in Ref.~\cite{Tan-Erk-Gio-Guh-2008},
bounds are effectively used for comparing
optimal success probabilities with different optical states.
In the case of minimum-error measurements,
Qiu {\it et al.} compared some of these upper bounds with each other,
and derived another upper bound \cite{Qiu-Li-2010}, which improves
some upper bounds in some cases.
In contrast, the square root measurement (SRM, also called the pretty good measurement),
is well known as a suboptimal measurement of the success probability criterion;
the success probability of the SRM is a good lower bound on the optimal one.
In the case of optimal inconclusive measurements,
an upper bound on the optimal success probability for binary quantum states
has been derived by Sugimoto {\it et al.} \cite{Sug-Has-Hor-Hay-2009}.

In the present study, new upper and lower bounds on the success probabilities
of minimum-error and optimal inconclusive measurements are derived.
The approach to this derivation exploits the fact that 
the optimal success probabilities are upper bounded by
suboptimal solutions to the dual problems of
optimal state discrimination problems.
We also present a necessary and sufficient condition for this new upper bound to be attainable,
from which the proposed lower bound can be obtained.
In the case of minimum-error measurements,
we show that a slightly modified version of the proposed bound
is tighter than Qiu {\it et al.}'s upper bound.
We also evaluate the performance of the proposed bounds through numerical experiments.
These experiments show that, on average, the proposed upper bound for minimum-error measurements
is tighter than Qiu {\it et al.}'s upper bound,
and the proposed bound for optimal inconclusive measurements
is tighter than Sugimoto {\it et al.}'s one in the case of binary quantum states.

\section{minimum-error and optimal inconclusive measurements}

We consider discrimination between $M$ quantum states
represented by a set of density operators $\{ \hsigma_m \}_{m \in \mI_M}$
with prior probabilities $\{ \xi_m \}_{m \in \mI_M}$, where $\mI_k = \{ 0, 1, \cdots, k-1 \}$.
$\hsigma_m$ satisfies $\hsigma_m \ge 0$ and $\Tr ~\hsigma_m = 1$,
where $\hA \ge 0$, $\hA \ge \hB$, and $\hA \le \hB$ respectively
denote that $\hA$, $\hA - \hB$, and $\hB - \hA$ are positive semidefinite.
To simplify notation, let $\hrho_m = \xi_m \hsigma_m$,
which we refer to as a quantum state.
We can easily verify $\hrho_m \ge 0$, $\Tr ~\hrho_m = \xi_m > 0$ for any $m \in \mI_M$,
and $\sum_{m=0}^{M-1} \Tr~\hrho_m = 1$.
A set of quantum states, $\rho = \{ \hrho_m \}_{m \in \mI_M}$, is referred to as a quantum state set.
Let $\mH$ be the state space of $\rho$, which is the Hilbert space spanned by the supports
of the operators $\{ \hrho_m \}$.

Let us consider a quantum measurement that may return an inconclusive answer,
which can be described by a positive operator-valued measure (POVM)
with $M+1$ detection operators, $\Pi = \{ \hPi_m \}_{m \in \mI_{M+1}}$.
The detection operator $\hPi_m$ with $m \in \mI_M$ corresponds to identification of
the state $\hrho_m$, while $\hPi_M$ corresponds to the inconclusive answer.
It is assumed without loss of generality that $\hPi_m$ is on $\mH$ for any $m \in \mI_{M+1}$.
Let $\POVM$ be the entire set of POVMs on $\mH$ each of which consists of $M+1$ detection operators;
then, any $\Pi \in \POVM$ satisfies
\begin{eqnarray}
 \hPi_m &\ge& 0, ~~~ \forall m \in \mI_{M+1}, \\ \nonumber
  \sum_{m=0}^M \hPi_m &=& \ident, \label{eq:POVM}
\end{eqnarray}
where $\ident$ is the identity operator on $\mH$.

The success probability, $\PC(\Pi)$, the error probability, $\PE(\Pi)$,
and the inconclusive probability, $\PI(\Pi)$, of a POVM $\Pi$ can be represented as
\begin{eqnarray}
 \PC(\Pi) &=& \sum_{m=0}^{M-1} \Tr(\hrho_m \hPi_m), \nonumber \\
 \PE(\Pi) &=& \underset{(m \neq k)}{\sum_{m=0}^{M-1} \sum_{k=0}^{M-1}}
  \Tr(\hrho_m \hPi_k), \nonumber \\
 \PI(\Pi) &=& \sum_{m=0}^{M-1} \Tr(\hrho_m \hPi_M) = \Tr (\hG\hPi_M),
\end{eqnarray}
where $\hG$ is the Gram operator of $\rho$ expressed as
\begin{eqnarray}
 \hG &=& \sum_{m=0}^{M-1} \hrho_m. \label{eq:G}
\end{eqnarray}
The sum of these probabilities is one, i.e.,
\begin{eqnarray}
 \PC(\Pi) + \PE(\Pi) + \PI(\Pi) = 1, \label{eq:PC_PE_PI}
\end{eqnarray}
for any $\Pi \in \POVM$.

An optimal inconclusive measurement $\Pi$ with
the inconclusive probability of $p$ $~(0 \le p \le 1)$ is a measurement maximizing
the success probability $\PC(\Pi)$ under the constraint that $\PI(\Pi) = p$;
i.e., it is an optimal solution to the following optimization problem:
\begin{eqnarray}
 \begin{array}{lll}
  {\rm P:} & {\rm maximize} & \PC(\Pi) \\
  & {\rm subject~to} & \Pi \in \POVM_p \\
 \end{array} \label{eq:inc_main}
\end{eqnarray}
with a POVM $\Pi$, where $\POVM_p$ is the entire set of POVMs, $\Pi \in \POVM$, satisfying $\PI(\Pi) = p$.
In particular, an optimal solution with $p = 0$ is called a minimum-error measurement,
which always satisfies $\hPi_M = 0$.
Let $\PCopt_p$ be the optimal value of problem~P,
i.e.,
\begin{eqnarray}
 \PCopt_p &=& \max_{\Pi \in \POVM_p} \PC(\Pi).
\end{eqnarray}
Also, let $\PCopt = \PCopt_0$, which is equal to the success probability of
a minimum-error measurement.

Problem~P is semidefinite programming,
and its dual problem can be represented as \cite{Eld-2003-inc}:
\begin{eqnarray}
 \begin{array}{lll}
  {\rm DP:} & {\rm minimize} & \Tr~\hZ - a p \\
  & {\rm subject~to} & \hZ \in \mS_a \\
 \end{array} \label{eq:inc_dual}
\end{eqnarray}
with a positive semidefinite operator $\hZ$ on $\mH$ and $a \in \Real_+$,
where $\Real_+$ is the entire set of nonnegative real numbers,
and $\mS_a$ is expressed as
\begin{eqnarray}
 \mS_a &=& \{ \hZ : \hZ \ge \hrho_m ~(\forall~m \in \mI_M), ~ \hZ \ge a \hG \}. \label{eq:inc_UB}
\end{eqnarray}
The optimal value of problem~DP is equal to
that of problem~P, i.e., $\PCopt_p$ \cite{Eld-2003-inc}.
The following inequality thus holds:
\begin{eqnarray}
 \Tr~\hZ - a p &\ge& \PCopt_p, ~~~ \forall a \in \Real_+, \hZ \in \mS_a. \label{eq:TrZap}
\end{eqnarray}
Similarly, the dual problem with $p = 0$ is represented as \cite{Eld-Meg-Ver-2003}:
\begin{eqnarray}
 \begin{array}{lll}
  {\mbox{\rm \DPme:}} & {\rm minimize} & \Tr~\hX \\
  & {\rm subject~to} & \hX \in \mS_0 \\
 \end{array} \label{eq:me_dual}
\end{eqnarray}
with a positive semidefinite operator $\hX$.
As in Eq.~(\ref{eq:TrZap}), we have
\begin{eqnarray}
 \Tr~\hX &\ge& \PCopt, ~~~ \forall \hX \in \mS_0. \label{eq:me_UB}
\end{eqnarray}

\section{Bounds on success probability of minimum-error measurement}

\subsection{Preparation}

Let the spectral decomposition of a Hermitian operator $\hA$ be
$\hA = \sum_n \lambda_n \hE_n$,
where $\lambda_n$ is an eigenvalue of $\hA$,
and $\hE_n$ is the corresponding projection operator.
Let $\hA_+$ be
\begin{eqnarray}
 \hA_+ &=& \sum_{\lambda_n > 0} \lambda_n \hE_n. \label{eq:A+}
\end{eqnarray}
Also, let $\Proja(\hA)$ and $\Projb(\hA)$, respectively, be
\begin{eqnarray}
 \Proja(\hA) &=& \sum_{\lambda_n > 0} \hE_n, ~ \Projb(\hA) = \sum_{\lambda_n \ge 0} \hE_n. \label{eq:Proja}
\end{eqnarray}
In other words, $\Proja(\hA)$ is the projection operator onto the support space of $\hA_+$,
and $\Projb(\hA)$ is the projection operator onto the kernel of $(-\hA)_+$.
From Eq.~(\ref{eq:Proja}), $\Projb(\hA) \ge \Proja(\hA)$ obviously holds.

In preparation for subsequent subsections,
we show the following lemma.

\begin{lemma} \label{lemma:AB_dual}
 Let $\hA$ and $\hB$ be positive semidefinite operators.
 We consider the following optimization problem
 \begin{eqnarray}
  \begin{array}{ll}
   {\rm minimize} & \Tr~\hY \\
   {\rm subject~to} & \hY \ge \hA, \hY \ge \hB \\
  \end{array} \label{eq:AB_dual}
 \end{eqnarray}
 with a variable $\hY$.
 Also, let $\hY^\opt = \hB + (\hA - \hB)_+$;
 accordingly, $\hY^\opt$ is the optimal solution to problem (\ref{eq:AB_dual}).
 In addition, any operator $\hPhi$ with $\ident \ge \hPhi \ge 0$ satisfies
 \begin{eqnarray}
  \Tr~\hY^\opt &\ge& \Tr(\hA \hPhi) + \Tr[\hB(\ident - \hPhi)]. \label{eq:AB_main}
 \end{eqnarray}
 The equality in Eq.~(\ref{eq:AB_main}) holds if and only if
 \begin{eqnarray}
  \Projb(\hA-\hB) \ge \hPhi \ge \Proja(\hA-\hB). \label{eq:AB_main_proj}
 \end{eqnarray}
\end{lemma}

\begin{proof}
The case of $\Tr(\hA + \hB) = 0$, i.e., $\hA = \hB = 0$, is obvious,
so we concentrate on $\Tr(\hA + \hB) \neq 0$.
Let $c = 1/\Tr(\hA + \hB)$, $\hrho_A = c\hA$, $\hrho_B = c\hB$, and $\hX = c\hY$;
then, problem (\ref{eq:AB_dual}) can be reformulated as
\begin{eqnarray}
  \begin{array}{ll}
   {\rm minimize} & \Tr~\hX \\
   {\rm subject~to} & \hX \ge \hrho_A, \hX \ge \hrho_B. \\
  \end{array} \label{eq:AB_dual_rho}
\end{eqnarray}
This is the dual problem of the problem of obtaining a minimum-error measurement
for a binary quantum state set $\{ \hrho_A, \hrho_B \}$.
Thus, the optimal solution is $\hX^\opt = \hrho_B + (\hrho_A - \hrho_B)_+$
(e.g., \cite{Hel-1969}).
Moreover, for any operator $\hPhi$ with $\ident \ge \hPhi \ge 0$,
$\{ \hPhi, \ident - \hPhi \}$ is a POVM for a binary quantum state set;
thus, it follows that
\begin{eqnarray}
  \Tr~\hX^\opt &\ge& \Tr(\hrho_A \hPhi) + \Tr[\hrho_B(\ident - \hPhi)]. \label{eq:AB_dual_TrX}
\end{eqnarray}
Dividing this equation by $c$ gives Eq.~(\ref{eq:AB_main}).
Obviously, the equality in (\ref{eq:AB_main}) holds
if and only if $\{ \hPhi, \ident - \hPhi \}$ is a minimum-error measurement,
i.e., (\ref{eq:AB_main_proj}) holds \cite{Hel-1969}.
\QED
\end{proof}

\subsection{Proposed upper bound}

According to Eq.~(\ref{eq:me_UB}), for any feasible solution to problem~\DPme, $\hX \in \mS_0$,
$\PCopt$ is upper bounded by $\Tr~\hX$.
Here, we consider obtaining a suboptimal solution to problem~\DPme
by using Lemma~\ref{lemma:AB_dual}.
For $m \in \mI_{M-1}$, the following optimization problem is considered:
\begin{eqnarray}
 \begin{array}{ll}
  {\rm minimize} & \Tr~\hX'_{m+1} \\
  {\rm subject~to} & \hX'_{m+1} \ge \hrho_{m+1}, \hX'_{m+1} \ge \hX_m \\
 \end{array} \label{eq:me_dual2}
\end{eqnarray}
with a positive semidefinite operator $\hX'_{m+1}$,
where $\hX_0 = \hrho_0$, and $\hX_{m+1}$ $~(m \in \mI_{M-1})$ is an optimal solution
to problem (\ref{eq:me_dual2}).
We derive a new upper bound on $\PCopt$, namely, $\PCUp = \Tr~\hX_{M-1}$.
According to Lemma~\ref{lemma:AB_dual}, the optimal solution to problem (\ref{eq:me_dual2})
is expressed as $\hX_{m+1} = \hX_m + (\hrho_{m+1} - \hX_m)_+$.
The proposed upper bound $\PCUp$ can thus be expressed as
\begin{eqnarray}
 \PCUp &=& \Tr~\hX_{M-1}, \nonumber \\
 \hX_0 &=& \hrho_0, \nonumber \\
 \hX_{m+1} &=& \hX_m + (\hrho_{m+1} - \hX_m)_+, ~~~ m \in \mI_{M-1}. \label{eq:me_UB_prop}
\end{eqnarray}

We can easily show that $\PCopt$ is upper bounded by $\PCUp$:
\begin{thm} \label{thm:UB}
 $\PCUp \ge \PCopt$.
\end{thm}

\begin{proof}
 From the constraint of problem (\ref{eq:me_dual2}),
 it is clear that $\hX_{M-1} \ge \hX_m \ge \hrho_m$ holds for any $m \in \mI_M$.
 Thus, $\hX_{M-1} \in \mS_0$ also holds,
 which gives $\PCUp \ge \PCopt$ from Eq.~(\ref{eq:me_UB}).
 \QED
\end{proof}

\begin{remark}
 For a set of binary states, $\PCUp = \PCopt$ holds.
\end{remark}

\begin{proof}
 Since $\hX_1 = \hrho_0 + (\hrho_1 - \hrho_0)_+$ is the optimal solution
 to problem~\DPme, $\PCUp = \Tr~\hX_1 = \PCopt$ holds.
 \QED
\end{proof}

In Ref.~\cite{Qiu-Li-2010}, Qiu {\it et al.} proposed an upper bound on
$\PCopt$, denoted as $\PCU{Qiu}$, expressed as
\begin{eqnarray}
 \PCU{Qiu} &=& \min_{k \in \mI_M} \PCU{Qiu}(k), \nonumber \\
 \PCU{Qiu}(k) &=& \xi_k + \sum_{\mI_M \ni m \neq k} \Tr(\hrho_m - \hrho_k)_+.
  \label{eq:me_UB_Qiu}
\end{eqnarray}
Note that $\PCU{Qiu}$ is identical to $1 - L_4$ in Ref.~\cite{Qiu-Li-2010}.
$\PCU{Qiu}(k)$ is equivalent to $\PCU{Qiu}(0)$ after permuting $\hrho_0$ and $\hrho_k$.
Here, we give a slightly modified version of $\PCUp$, denoted as $\PCUp'$,
and show $\PCUp' \le \PCU{Qiu}$.
$\PCUp'$ is defined as
\begin{eqnarray}
 \PCUp' = \min_{k \in \mI_M} \PCUp(k), \label{eq:PCUp'}
\end{eqnarray}
where $\PCUp(k)$ is $\PCUp$ obtained from Eq.~(\ref{eq:me_UB_prop})
after permuting $\hrho_0$ and $\hrho_k$.
Since $\PCUp(k) \ge \PCopt$ holds for any $k \in \mI_M$,
$\PCopt$ is obviously upper bounded by $\PCUp'$.
Moreover, from $\PCUp(0) = \PCUp$, $\PCUp' \le \PCUp$ holds.
The following proposition also holds:

\begin{proposition} \label{pro:Qiu}
 $\PCUp' \le \PCU{Qiu}$.
\end{proposition}

\begin{proof}
It suffices to show $\PCUp(k) \le \PCU{Qiu}(k)$ for any $k \in \mI_M$.
Since $\PCUp(k) \le \PCU{Qiu}(k)$ is equivalent to $\PCUp(0) \le \PCU{Qiu}(0)$
for the quantum state set that is obtained by permutation of $\hrho_0$ and $\hrho_k$,
it is only necessary to show $\PCUp(0) \le \PCU{Qiu}(0)$ for any quantum state set.
Since $\hX_m \ge \hrho_0$ gives $\hrho_{m+1} - \hrho_0 \ge \hrho_{m+1} - \hX_m$
for any $m \in \mI_{M-1}$,
from Lemma~\ref{lemma:AB} in Appendix~\ref{append:Hermite},
\begin{eqnarray}
 \Tr(\hrho_{m+1} - \hrho_0)_+ &\ge& \Tr(\hrho_{m+1} - \hX_m)_+
\end{eqnarray}
is obtained.
Therefore, Eqs.~(\ref{eq:me_UB_prop}) and (\ref{eq:me_UB_Qiu}) give
\begin{eqnarray}
 \PCUp(0) &=& \Tr~\hrho_0 + \sum_{m=0}^{M-2} \Tr~(\hrho_{m+1} - \hX_m)_+ \nonumber \\
 &\le& \xi_0 + \sum_{m=0}^{M-2} \Tr~(\hrho_{m+1} - \hrho_0)_+ = \PCU{Qiu}(0).
\end{eqnarray}
\QED
\end{proof}

\subsection{Attainability of proposed upper bound} \label{subsec:me_nas}

A necessary and sufficient condition for the proposed upper bound
to achieve the optimal success probability is provided by the following theorem:

\begin{thm} \label{thm:nas}
$\PCUp = \PCopt$ holds if and only if
$\{ \hE_k \}_{k=1}^{M-1}$ exists such that
\begin{eqnarray}
  \Projb[\hA_k(\hX_{k-1} - \hrho_k)\hA_k^\dag] &\ge& \hE_k \nonumber \\
  &\ge& \Proja[\hA_k(\hX_{k-1} - \hrho_k)\hA_k^\dag], \nonumber \\
  \lefteqn{ ~~~ k \in \{ 1, \cdots, M-1 \}, } \label{eq:nas_Pi_cond}
\end{eqnarray}
and
\begin{eqnarray}
  \hA_m (\hX_{m-1} - \hrho_m)_+ \hA_m^\dag &=& [ \hA_m (\hX_{m-1} - \hrho_m) \hA_m^\dag ]_+,
   \nonumber \\
  \lefteqn{ m \in \{ 1, 2, \cdots, M-2 \}, } \label{eq:nas_commute}
\end{eqnarray}
where
\begin{eqnarray}
  \hA_m &=&
   \left\{
       \begin{array}{ll}
       \hE_{m+1}^\half \hE_{m+2}^\half \cdots \hE_{M-1}^\half, & 0 \le m < M - 1, \\
       \ident, & m = M-1. \label{eq:nas_Am}
       \end{array} \right.
\end{eqnarray}
\end{thm}

\begin{proof}
In preparation for the proof, a set of operators, $\Pi = \{ \hPi_m \}_{m \in \mI_M}$, is defined as
\begin{eqnarray}
  \hPi_m &=&
   \left\{
       \begin{array}{ll}
       |\hA_m|^2 - |\hA_{m-1}|^2, & 1 \le m \le M - 1, \\
       |\hA_0|^2, & m = 0,
       \end{array} \right. \label{eq:Pi_A}
\end{eqnarray}
where $|\hA| = (\hA^\dag\hA)^{1/2}$.
For any $\{ \hE_k \}_{k=1}^{M-1}$ with $\ident \ge \hE_k \ge 0$,
\begin{eqnarray}
  \sum_{m=0}^{M-1} \hPi_m &=& |\hA_{M-1}|^2 = \ident, \nonumber \\
  \hPi_m &=& \hA_m^\dag(\ident - \hE_m)\hA_m \ge 0, ~~~ 1 \le m \le M-1, \nonumber \\
  \hPi_0 &=& |\hA_0|^2 \ge 0 \label{eq:nas_sum_Pi}
\end{eqnarray}
holds.
The second line of Eq.~(\ref{eq:nas_sum_Pi}) follows from $|\hA_{m-1}|^2 = \hA_m^\dag\hE_m\hA_m$,
which is given by Eq.~(\ref{eq:nas_Am}).
Thus, $\Pi$ is a POVM.
On the contrary, for any POVM $\Pi = \{ \hPi_m \}$,
$\{ \hE_k \}_{k=1}^{M-1}$ exists such that
$\ident \ge \hE_k \ge 0$ and Eq.~(\ref{eq:Pi_A}) hold (see Appendix~\ref{append:Pi_E}).

In the following, $\{ \hE_k \}$ satisfying $\ident \ge \hE_k \ge 0$ $~(1 \le k \le M-1)$
and its corresponding POVM $\Pi$, defined by Eq.~(\ref{eq:Pi_A}), are considered.
From Lemma~\ref{lemma:AHA_ge} in Appendix~\ref{append:Hermite} and
$\hX_m = \hX_{m-1} + (\hrho_m - \hX_{m-1})_+$,
it follows that for any $m$ with $1 \le m \le M-1$,
\begin{eqnarray}
  \lefteqn{ \Tr(\hA_m\hX_m\hA_m^\dag) } \nonumber \\
  &\ge& \Tr[\hA_m\hrho_m\hA_m^\dag(\ident - \hE_m)]
   + \Tr(\hA_m\hX_{m-1}\hA_m^\dag \hE_m) \nonumber \\
  &=& \Tr(\hrho_m\hPi_m) + \Tr(\hA_{m-1}\hX_{m-1}\hA_{m-1}^\dag), \label{eq:nas_n_AXA}
\end{eqnarray}
where the last line follows from
$\hA_m^\dag(\ident - \hE_m)\hA_m = \hPi_m$ and
$\hA_m^\dag\hE_m\hA_m = |\hA_{m-1}|^2$.
Using Eq.~(\ref{eq:nas_n_AXA}) recursively for $m = M-1, M-2, \cdots, 1$ yields
\begin{eqnarray}
  \PCUp &=& \Tr~\hX_{M-1} \nonumber \\
  &\ge& \sum_{m=1}^{M-1} \Tr(\hrho_m\hPi_m) + \Tr(\hrho_0|\hA_0|^2) \nonumber \\
  &=& \sum_{m=0}^{M-1} \Tr(\hrho_m\hPi_m) = \PC(\Pi). \label{eq:nas_n_PCUp}
\end{eqnarray}

First, we prove the sufficiency of Theorem~\ref{thm:nas}.
Assume $\PCUp = \PCopt$.
$\Pi = \{ \hPi_m \}$ is taken as a minimum-error measurement.
$\hE_m$ is chosen to satisfy $\ident \ge \hE_m \ge 0$ and Eqs.~(\ref{eq:nas_Am}) and (\ref{eq:Pi_A}).
Then, from $\PCUp = \PCopt = \PC(\Pi)$, the equality in Eq.~(\ref{eq:nas_n_PCUp}) holds,
implying that the equality in Eq.~(\ref{eq:nas_n_AXA}) holds for any
$m \in \{ M-1,M-2,\cdots,1 \}$.
Therefore, according to Lemma~\ref{lemma:AHA_ge},
Eqs.~(\ref{eq:nas_Pi_cond}) and (\ref{eq:nas_commute}) hold.

Next, we prove the necessity of Theorem~\ref{thm:nas}.
Assume that $\{ \hE_k \}_{k=1}^{M-1}$ exists such that
Eqs.~(\ref{eq:nas_Pi_cond}) and (\ref{eq:nas_commute}) hold.
Also, let $\Pi = \{ \hPi_m \}$ be the POVM defined by Eq.~(\ref{eq:Pi_A}).
According to Lemma~\ref{lemma:AHA_ge},
the equality in Eq.~(\ref{eq:nas_n_AXA}) holds for any $m \in \{ M-1,M-2,\cdots,1 \}$;
thus, the equality in Eq.~(\ref{eq:nas_n_PCUp}), i.e. $\PCUp = \PC(\Pi)$, holds.
From $\PCUp \ge \PCopt \ge \PC(\Pi)$, $\PCUp = \PCopt$ therefore also holds.
\QED
\end{proof}

$\he_m$ and $\ha_m$ are defined as
\begin{eqnarray}
 \he_m &=& \Projb[\ha_m(\hX_{m-1} - \hrho_m)\ha_m^\dag], \nonumber \\
 \ha_m &=&
  \left\{
   \begin{array}{ll}
	\he_{m+1} \he_{m+2} \cdots \he_{M-1}, & 0 \le m < M - 1, \\
	\ident, & m = M-1. \label{eq:nas_am}
   \end{array} \right.
\end{eqnarray}
Note that if $\hE_m = \he_m$, then $\hA_m = \ha_m$.
The following corollary (proof in Appendix~\ref{append:cor_nas_supp}) holds:

\begin{cor} \label{cor:nas_supp}
 Assume that, for any $m$ with $1 \le m \le M-1$,
 \begin{eqnarray}
  \supp[\ha_m(\hX_{m-1} - \hrho_m)\ha_m^\dag] &=& \supp~\ha_m\hX_m\ha_m^\dag. \label{eq:suppX}
 \end{eqnarray}
 Then, $\PCUp = \PCopt$ holds if and only if
 \begin{eqnarray}
  \ha_m (\hX_{m-1} - \hrho_m)_+ \ha_m^\dag &=& [ \ha_m (\hX_{m-1} - \hrho_m) \ha_m^\dag ]_+,
   \nonumber \\
  \lefteqn{ m \in \{ 1, 2, \cdots, M-2 \}. } \label{eq:nas_commute2}
 \end{eqnarray}
\end{cor}

%Corollary~\ref{cor:nas_supp} may be useful for some quantum state sets;
%in particular, Eq.~(\ref{eq:suppX}) holds if
%the supports of the operators $\{ \hrho_m \}$ are linearly independent.
%Indeed, in this case, since $\supp~\hX_k = \supp(\hX_{k-1} - \hrho_k)$ holds,
%from Lemma~\ref{lemma:supp} in Appendix~\ref{append:Hermite}, Eq.~(\ref{eq:suppX}) holds.

\subsection{Proposed lower bound} \label{subsec:me_lower}

The proof of Theorem~\ref{thm:nas} shows that
if $\PCUp = \PCopt$, then the POVM $\{ \hPi_m \}_{m \in \mI_M}$ of Eq.~(\ref{eq:Pi_A}),
which is obtained from the corresponding $\{ \hE_k \}_{k=1}^{M-1}$,
is a minimum-error measurement.
In particular, substituting $\hE_k = \he_k$ gives that
the POVM $\Pi^\circ = \{ \hPi^\circ_m \}_{m \in \mI_M}$ defined as
\begin{eqnarray}
 \hPi^\circ_m &=&
  \left\{
   \begin{array}{ll}
	|\ha_m|^2 - |\ha_{m-1}|^2, & 0 < m \le M - 1, \\
	|\ha_0|^2, & m = 0,
   \end{array} \right.
\end{eqnarray}
where $\he_m$ and $\ha_m$ are given by Eq.~(\ref{eq:nas_am}).
is also a minimum-error measurement when $\PCUp = \PCopt$.
Exploiting this fact, we propose a lower bound on $\PCopt$, denoted as $\PCLp$,
expressed as
\begin{eqnarray}
\PCLp &=& \PC(\Pi^\circ) = \sum_{m=0}^{M-1} \Tr(\hrho_m\hPi^\circ_m). \label{eq:PCLp}
\end{eqnarray}
Since $\Pi^\circ$ is a POVM, $\PCLp \le \PCopt$ obviously holds.
The SRM $\Pi^\SRM = \{ \hPi^\SRM_m \}_{m \in \mI_M}$, which is defined as
\begin{eqnarray}
 \hPi^\SRM_m &=& \hG^{-\frac{1}{2}} \hrho_m \hG^{-\frac{1}{2}},
\end{eqnarray}
is well known as a good approximation to a minimum-error measurement.
We will show in numerical experiments in Section~\ref{sec:examples} that 
$\PCLp$ tends to be closer to $\PCopt$ than the success probability of the SRM.

\section{Bounds on success probability of optimal inconclusive measurement}

\subsection{Proposed upper bound}

The arguments presented in the previous section can be extended to optimal inconclusive measurements
as follows.
Assume that a suboptimal solution, $\hX^\sopt$, to problem~\DPme for a quantum state set $\rho$
is given.
In this paper, let $\hX^\sopt = \hX_{M-1}$, which is defined by Eq.~(\ref{eq:me_UB_prop}).
Note that if an optimal solution $\hX^\opt$ to problem~\DPme is given,
then $\hX^\sopt = \hX^\opt$ can be used instead of $\hX^\sopt = \hX_{M-1}$.
A suboptimal solution to problem~DP can be obtained by
solving the following optimization problem:
\begin{eqnarray}
 \begin{array}{ll}
  {\rm minimize} & \Tr~\hZ - a p \\
  {\rm subject~to} & \hZ \in \mZ_a \\
 \end{array} \label{eq:inc_dual2t}
\end{eqnarray}
with a positive semidefinite operator $\hZ$ on $\mH$ and $a \in \Real_+$,
where
\begin{eqnarray}
 \mZ_a &=& \{ \hZ : \hZ \ge a \hG, \hZ \ge \hX^\sopt \}.
\end{eqnarray}
Indeed, from $\hX^\sopt \in \mS_0$, $\hZ \in \mS_a$ holds for any $\hZ \in \mZ_a$;
i.e., $\hZ$ is a feasible solution to problem~DP.
Accordingly, $\PCopt_p$ is upper bounded by the optimal value of problem (\ref{eq:inc_dual2t}).
Let
\begin{eqnarray}
 s(a) &=& \min_{\hZ \in \mZ_a} \Tr~\hZ - a p; \label{eq:sa0}
\end{eqnarray}
then, the optimal value of problem (\ref{eq:inc_dual2t}) is equal to
$\min_{a \in \Real_+} s(a)$.
Lemma~\ref{lemma:AB_dual} indicates that
$\Tr~\hZ \ge \Tr~\hX^\sopt + \Tr(a \hG - \hX^\sopt)_+$ holds for any $\hZ \in \mZ_a$
and the equality holds when $\hZ = \hX^\sopt + (a \hG - \hX^\sopt)_+$.
Thus, we have
\begin{eqnarray}
 s(a) &=& \Tr~\hX^\sopt + \Tr(a \hG - \hX^\sopt)_+ - a p. \label{eq:sa}
\end{eqnarray}

Since it is difficult to obtain the optimal value, $\min_{a \in \Real_+} s(a)$,
of problem (\ref{eq:inc_dual2t}) in general,
we consider computing the minimum $s(a)$ for several values of $a$
as a suboptimal solution.
We propose an upper bound on $\PCopt_p$, denoted as $\PCUip$, expressed as
\begin{eqnarray}
 \PCUip &=& \min \left\{ 1 - p, \min_{a \in \mA} s(a) \right\}, \label{eq:PCUip}
\end{eqnarray}
where $\mA \subseteq \Real_+$ is a set of candidates for $a$.
Note that, from Eq.~(\ref{eq:PC_PE_PI}), $\PCopt_p \le 1 - p$ always holds, and
Eq.~(\ref{eq:PCUip}) guarantees that $\PCUip$ does not exceed $1 - p$.
It is expected that $\PCUip$ can be effectively obtained
by adaptively selecting appropriate candidates.

\begin{thm} \label{thm:UB_inc}
 $\PCUip \ge \PCopt_p$.
\end{thm}

\begin{proof}
Since the case of $\PCUip = 1 - p$ is obvious, we assume $\PCUip < 1 - p$.
Recall that $\hZ \in \mS_a$ holds for any $\hZ \in \mZ_a$.
Thus, Eqs.~(\ref{eq:TrZap}) and (\ref{eq:sa0}) give
\begin{eqnarray}
 s(a) &=& \min_{\hZ \in \mZ_a} \Tr~\hZ - a p \ge \min_{\hZ \in \mS_a} \Tr~\hZ - a p \ge \PCopt_p.
\end{eqnarray}
Therefore, from Eq.~(\ref{eq:PCUip}), we have
\begin{eqnarray}
  \PCUip &=& \min_{a \in \mA} s(a) \ge \PCopt_p.
\end{eqnarray}
\QED
\end{proof}

Algorithm~1 shows an example of computing $\PCUip$.
We will provide a concrete algorithm on how to initialize and update $a$
in Subsection \ref{subsec:inc_algorithm}.
\begin{figure}
 \begin{algorithm}[H]
  \caption{An example of computing $\PCUip$.}
  \begin{algorithmic}[1]
   \REQUIRE $\{ \hrho_m \}_{m \in \mI_M}$, $p$
   \STATE Let $\hX^\sopt = \hX_{M-1}$, where $\hX_{M-1}$ is given by Eq.~(\ref{eq:me_UB_prop})
   \STATE $\PCUip \leftarrow 1 - p$
   \STATE Initialize $a$
   \FOR{$j \leftarrow 1,2,\cdots$}
   \STATE Compute $s(a)$ from Eq.~(\ref{eq:sa})
   \STATE $\PCUip \leftarrow \min \{ \PCUip, s(a) \}$
   \STATE Update $a$
   \ENDFOR
   \ENSURE $\PCUip$
  \end{algorithmic}
 \end{algorithm}
\end{figure}

\subsection{Properties of $s(a)$} \label{subsec:inc_sa}

To appropriately update $a$ in Algorithm~1,
the properties of $s(a)$ should be well understood.
The following proposition shows some of the properties (proof in Appendix~\ref{append:g}):
\begin{proposition} \label{pro:g}
Let $\lambda_{\max}(\hA)$ and $\lambda_{\min}(\hA)$ be
the maximum and minimum eigenvalues of a positive semidefinite operator $\hA$, respectively.
$s(a)$ satisfies the following conditions:
\begin{enumerate}[(1)]
 \setlength{\parskip}{0cm}
 \setlength{\itemsep}{0cm}
 \item If $a \le \lambda_{\min}(\hG^{-1/2} \hX^\sopt \hG^{-1/2})$, then $s(a) = \Tr~\hX^\sopt - a p$ holds.
	   Also, $1/M \le \lambda_{\min}(\hG^{-1/2} \hX^\sopt \hG^{-1/2})$ holds.
 \item If $a \ge \lambda_{\max}(\hG^{-1/2} \hX^\sopt \hG^{-1/2})$, then $s(a) = a (1 - p)$ holds.
 \item $s(a)$ is convex with respect to $a$.
\end{enumerate}
Note that since $\hG$ is a positive definite operator on $\mH$, $\hG^{-1/2}$ exists.
\end{proposition}

The following proposition also holds (proof in Appendix~\ref{append:sa}):

\begin{proposition} \label{pro:sa}
Let $\tp(a) = \Tr[\hG \Proja(a\hG - \hX^\sopt)]$ and $\tp^+(a) = \Tr[\hG \Projb(a\hG - \hX^\sopt)]$;
then, the following conditions hold:
\begin{enumerate}[(1)]
 \setlength{\parskip}{0cm}
 \setlength{\itemsep}{0cm}
 \item If $a < a'$, then $\tp^+(a) \le \tp(a')$ holds.
	   In addition, $\tp(a)$ and $\tp^+(a)$ monotonically increase with respect to $a$.
 \item $a$ minimizes $s(a)$ if and only if $\tp(a) \le p \le \tp^+(a)$ holds.
\end{enumerate}
\end{proposition}

\subsection{Algorithm for computing proposed upper bound} \label{subsec:inc_algorithm}

Propositions~\ref{pro:g} and \ref{pro:sa} are useful to update $a$ in Algorithm~1.
For example, since $\tp(a)$ monotonically increases with respect to $a$,
as stated in Proposition~\ref{pro:sa},
$a$ should be updated to a larger value if $\tp(a) < p$
or the smaller value if $\tp(a) > p$.

A concrete example of Algorithm~1 is shown in Algorithm~2.
Let $a^\opt \in \argmin_a s(a)$.
When initializing and updating $a$, Algorithm~2 exploits
Propositions~\ref{pro:g} and \ref{pro:sa}.
%Proposition~\ref{pro:sa} 'æ'èC$\tp(a^\opt) \le p \le \tp^+(a^\opt)$ 'ð–ž'½'·B
In steps~4 and 7, $a_\L$ and $a_\R$ are respectively initialized to
$\lambda_{\min}(\hG^{-1/2}\hX^\sopt\hG^{-1/2})$ and $\lambda_{\max}(\hG^{-1/2}\hX^\sopt\hG^{-1/2}) + \epsilon$,
where $\epsilon$ is a sufficiently small positive number.
Accordingly, since $a_\L\hG - \hX^\sopt \le 0$ and $a_\R\hG - \hX^\sopt \ge \epsilon\hG$ hold
(see Eqs.~(\ref{eq:XaG}) and (\ref{eq:XaG2}) in Appendix~\ref{append:g}),
$\tp(a_\L) = 0$ and $\tp(a_\R) = 1$ hold.
Thus, from $\tp(a_\L) \le p \le \tp(a_\R)$ and Proposition~\ref{pro:sa},
$a_\L \le a^\opt \le a_\R$ holds.
%It also follows from the monotonicity of $\tp(a)$ that
%$\tp(a_\L) \le p \le \tp(a_\R)$ holds.
In step~11, an estimated $a^\opt$, i.e., $a$, is computed on the assumption that
$\tp(a')$ is well approximated as linear in $a_\L \le a' \le a_\R$;
such $a$ satisfies $a_\L \le a \le a_\R$.
In steps~14--18, $a$ is substituted into $a_\L$ if $\tp(a) \le p$ (i.e., $a \le a^\opt$);
otherwise, $a$ is substituted into $a_\R$.
As a result, steps~10--19 guarantee that $a_\L$ and $a_\R$ satisfy $a_\L \le a^\opt \le a_\R$
and are closer to $a^\opt$ than those in the previous iteration.
The iteration process in Algorithm~2 stops after a fixed number of iterations;
alternatively, it may continue until certain stopping criteria
(e.g., the difference between $a_\L$ and $a_\R$ is sufficiently small) are met.
It is obvious that the difference between $\PCUip$ and $\PCopt_p$
monotonically decreases as the number of iterations, $J$, increases.

\begin{figure}
\begin{algorithm}[H]
\caption{Concrete example of computing $\PCUip$.}
\begin{algorithmic}[1]
\REQUIRE $\{ \hrho_m \}_{m \in \mI_M}$, $p$
\STATE Let $\hX^\sopt = \hX_{M-1}$, where $\hX_{M-1}$ is given by Eq.~(\ref{eq:me_UB_prop})
\STATE $\PCUip \leftarrow 1 - p$
\STATE /* Initialize $a_\L$ */
\STATE $a_\L \leftarrow \lambda_{\min}(\hG^{-1/2}\hX^\sopt\hG^{-1/2})$
\STATE $\PCUip \leftarrow \min \{ \PCUip, \Tr~\hX^\sopt - a_\L p \}$
\STATE /* Initialize $a_\R$ */
\STATE $a_\R \leftarrow \lambda_{\max}(\hG^{-1/2}\hX^\sopt\hG^{-1/2}) + \epsilon$
\STATE $\PCUip \leftarrow \min \{ \PCUip, a_\R (1 - p) \}$
\STATE /* Iterate */
\FOR{$j \leftarrow 1,2,\cdots,J$}
\STATE $a \leftarrow [[\tp(a_\R) - p] a_\L + [p - \tp(a_\L)] a_\R] / [\tp(a_\R) - \tp(a_\L)]$
\STATE Compute $s(a)$ using Eq.~(\ref{eq:sa})
\STATE $\PCUip \leftarrow \min \{ \PCUip, s(a) \}$
\IF{$\tp(a) \le p$}
\STATE $a_\L \leftarrow a$
\ELSE
\STATE $a_\R \leftarrow a$
\ENDIF
\ENDFOR
\ENSURE $\PCUip$
\end{algorithmic}
\end{algorithm}
\end{figure}

\subsection{Attainability of proposed upper bound} \label{subsec:inc_nas}

A necessary and sufficient condition for $\PCUip = \PCopt_p$ is determined as follows.
First, $a^\opt$ is taken as the optimal solution of $a$ in problem~DP.
Then, we consider solving the following optimization problem:
\begin{eqnarray}
\begin{array}{ll}
  {\rm minimize} & \Tr~\hZ \\
  {\rm subject~to} & \hZ \in \mS_{a^\opt} \\
\end{array} \label{eq:inc_dual_fixed_a}
\end{eqnarray}
with $\hZ$.
Since the optimal value of problem~DP is $\PCopt_p$,
the optimal value of problem (\ref{eq:inc_dual_fixed_a}) is $\PCopt_p + a^\opt p$.
Comparing Eqs.~(\ref{eq:me_dual}) and (\ref{eq:inc_dual_fixed_a})
indicates that Eq.~(\ref{eq:inc_dual_fixed_a}) can be regarded as
the problem of finding the success probability of a minimum-error measurement
for the set of $M+1$ quantum states $\rho' = \{ c \hrho_m \}_{m \in \mI_{M+1}}$,
with $\hrho_M = a^\opt \hG$,
where $c = 1 / (1 + a^\opt)$ is a constant such that $\sum_{m \in \mI_{M+1}} \Tr(c \hrho_m) = 1$.
Therefore,
% the same discussion as in Subsection~\ref{subsec:me_nas} applies here;
Theorem~\ref{thm:nas} and Corollary~\ref{cor:nas_supp} can be applied
in the case of optimal inconclusive measurements.

\subsection{Proposed lower bound} \label{subsec:inc_lower}

It is easy to extend the discussion in Subsection~\ref{subsec:me_lower}
to optimal inconclusive measurements.
Assume that $a_\L$ and $a_\R$ satisfying $\tp(a_\L) \le p \le \tp(a_\R)$ are given
(such $a_\L$ and $a_\R$ can be obtained from Algorithm~2).
$\Pi^{(a)} = \{ \hPi^{(a)}_m \}_{m \in \mI_{M+1}}$ is defined as
\begin{eqnarray}
 \hPi_m^{(a)} &=&
  \left\{
   \begin{array}{ll}
       |\ha_m^{(a)}|^2 - |\ha_{m-1}^{(a)}|^2, & 0 < m \le M, \\
       |\ha_0^{(a)}|^2, & m = 0, \\
   \end{array} \right. \nonumber \\
 \ha_m^{(a)} &=&
  \left\{
   \begin{array}{ll}
       \he^{(a)}_{m+1} \he^{(a)}_{m+2} \cdots \he^{(a)}_M, & 0 \le m < M, \\
       \ident, & m = M,
   \end{array} \right. \nonumber \\
 \he^{(a)}_m &=&
  \left\{
   \begin{array}{ll}
       \Projb[\ha_m(\hX_{m-1} - \hrho_m)\ha_m^\dag], & 0 < m \le M - 1, \\
       \Projb(\hX_{M-1} - a \hG), & m = M. \\
   \end{array} \right. \nonumber \\
 \label{eq:inc_lower_Pi}
\end{eqnarray}
Therefore, as discussed in Subsection~\ref{subsec:me_lower},
it is clear that $\Pi^{(a)}$ is a POVM.
In addition, since
\begin{eqnarray}
 \hPi_M^{(a)} &=& |\ha_M^{(a)}|^2 - |\ha_{M-1}^{(a)}|^2
  = \ident - \Projb(\hX_{M-1} - a \hG) \nonumber \\
 &=& \Proja(a \hG - \hX_{M-1})
\end{eqnarray}
holds, the inconclusive probability of the POVM $\Pi^{(a)}$
can be formulated as
\begin{eqnarray}
 \Tr(\hG\hPi_M^{(a)}) &=& \Tr[\hG\Proja(a \hG - \hX_{M-1})] = \tp(a).
\end{eqnarray}
Let us consider the POVM $\Pi^\bullet = \{ \hPi^\bullet_m \}_{m \in \mI_{M+1}}$,
where $\Pi^\bullet$ is defined as
\begin{eqnarray}
 \hPi^\bullet_m &=& \frac{[\tp(a_\R) - p] \hPi_m^{(a_\L)} + [p - \tp(a_\L)] \hPi_m^{(a_\R)}}{\tp(a_\R) - \tp(a_\L)}
  \label{eq:inc_lower_Pi_bullet}
\end{eqnarray}
if $\tp(a_\R) \neq \tp(a_\L)$, $\hPi^\bullet_m = \hPi_m^{(a_\L)}$ otherwise.
It is easy to verify that $\PI(\Pi^\bullet) = p$ holds.
We use the success probability of $\Pi^\bullet$, $\PC(\Pi^\bullet)$,
as a lower bound on $\PCopt_p$, denoted as $\PCLip$;
i.e., $\PCLip$ is given by
\begin{eqnarray}
 \PCLip &=& \PC(\Pi^\bullet) = \sum_{m=0}^{M-1} \Tr(\hrho_m\hPi^\bullet_m). \label{eq:PCLip}
\end{eqnarray}
From $\PI(\Pi^\bullet) = p$, $\PCLip \le \PCopt_p$ obviously holds.

\section{Computational complexity}

In this section, we discuss the computational complexity of
computing the proposed bounds.

First, the computational complexity of
computing $\PCUp$ and $\PCLp$ is investigated.
With regard to $\PCUp$, which is computed from Eq.~(\ref{eq:me_UB_prop}),
the major computational cost is computing $(\hrho_{m+1} - \hX_m)_+$.
It can be derived by computing the eigenvalues and their
corresponding eigenvectors of $\hrho_{m+1} - \hX_m$
and then using Eq.~(\ref{eq:A+}).
Let $N = \dim~\mH$.
The computation of the eigenvalues and eigenvectors
generally takes $O(N^3)$ time,
which indicates that the time complexity required by computing $\PCUp$ is $O[(M-1)N^3]$.
(Similarly, the time complexity of computing $\PCUp'$ in Eq.~(\ref{eq:PCUp'}) is $O[M(M-1)N^3]$.)
In contrast, the computation of $\PCU{Qiu}$ in Eq.~(\ref{eq:me_UB_Qiu})
requires $O[M(M-1)N^3]$ time,
which is $O(M)$ times longer than that for computing $\PCUp$.
Although $\PCUp$ is not always tighter than $\PCU{Qiu}$,
the numerical results presented in the next section demonstrate that
$\PCUp < \PCU{Qiu}$ holds on average.
With regard to $\PCLp$,
it is assumed that $X_m$ ~($m \in \mI_M$) in Eq.(\ref{eq:me_UB_prop}) is given;
from Eqs.~(\ref{eq:nas_am}) and (\ref{eq:PCLp}),
the major computational cost is computing $\Projb(\cdot)$ and
operator multiplication.
Both of them generally require $O(N^3)$ time,
and thus the computation of $\PCLp$ takes $O[(M-1)N^3]$.
Note that Ref.~\cite{Car-Vig-2010} provides a method of
computing the eigenvalues and eigenvectors of $\hrho_{m+1} - \hX_m$
from those of a corresponding $(\rank~\hrho_{m+1} + \rank~\hX_m)$-dimensional
square matrix;
this method can reduce the cost of computing $\PCUp$ and $\PCLp$
if $\rank~\hrho_{m+1} + \rank~\hX_m$ is smaller than $N$.

Next, the computational complexity of
computing $\PCUip$ and $\PCLip$ is investigated.
With regard to $\PCUip$, which is computed by Algorithm~2,
the major computational cost is computing
the following values:
(a) $\hX^\sopt$ in step~1,
(b) $\lambda_{\min}(\hG^{-1/2}\hX^\sopt\hG^{-1/2})$ in step~4
and $\lambda_{\max}(\hG^{-1/2}\hX^\sopt\hG^{-1/2})$ in step~7,
and (c) $s(a)$ in step~12 and $\tp(a)$ in step~14.
Since the computational complexities of
computing the $(-1/2)$-th power of an operator,
operator multiplication, and the eigenvalues and eigenvectors
are all $O(N^3)$,
the computations of values (a)--(c) above respectively
require $O[(M-1)N^3]$, $O(N^3)$, and $O(JN^3)$ times.
Therefore, the total computational complexity of computing $\PCUip$
is roughly $O[(M+J)N^3]$;
in particular, in the case of $M \gg J$, it is close to
that of computing $\PCUp$.
With regard to $\PCLip$,
we can make a similar discussion of $\PCLp$.
Assume that $X_m$ ~($m \in \mI_M$) in Eq.(\ref{eq:me_UB_prop}) is given.
From Eqs.~(\ref{eq:inc_lower_Pi}), (\ref{eq:inc_lower_Pi_bullet}),
and (\ref{eq:PCLip}),
the major computational cost is computing $\Projb(\cdot)$ and
operator multiplication, both of which generally take $O(N^3)$ time.
Thus, the total computational complexity of $\PCLip$ is $O(MN^3)$.

\section{Numerical examples} \label{sec:examples}

We discuss the accuracy of the proposed bounds on the success probabilities
of minimum-error and optimal inconclusive measurements
through numerical examples as follows.

One-hundred sets of randomly generated $M$ quantum states,
$\rho = \{ \hrho_m \}_{m \in \mI_M}$ with $\rank~\hrho_m = R$ $~(m \in \mI_M)$,
where $M$ and $R$ are parameters, were used in these examples.
Prior probabilities were also randomly selected.
The optimal success probability $\PCopt_p$ and 
the average relative errors between an upper or lower bound,
which is defined as $|\PCUip - \PCopt_p| / \PCopt_p$ or $|\PCLip - \PCopt_p| / \PCopt_p$,
were computed.
In the case of optimal inconclusive measurements,
the inconclusive probability, $p$, was randomly selected
in the range from 0 to 0.2.

\subsection{Case of minimum-error measurements}

Figure~\ref{fig:result_me_UB} shows the average relative errors of
the proposed upper bound, $\PCUp$, and Qiu {\it et al.}'s upper bound,
$\PCU{Qiu}$.
We observed that, at least in the range of $3 \le M \le 9$ and $R \le 9$,
the average relative error of $\PCUp$ is more than eight times
smaller than that of $\PCU{Qiu}$,
while $\PCUp < \PCU{Qiu}$ is not guaranteed for each quantum state set.
It also shows that the average relative error of $\PCUp$
increases gradually with increasing $M$,
while that of $\PCU{Qiu}$ increases rapidly.
Note that, in the case of $M = 2$, the average relative errors of
$\PCUp$ and $\PCU{Qiu}$ are always zero.

\begin{figure}[tb]
\centering
\includegraphics[scale=0.8]{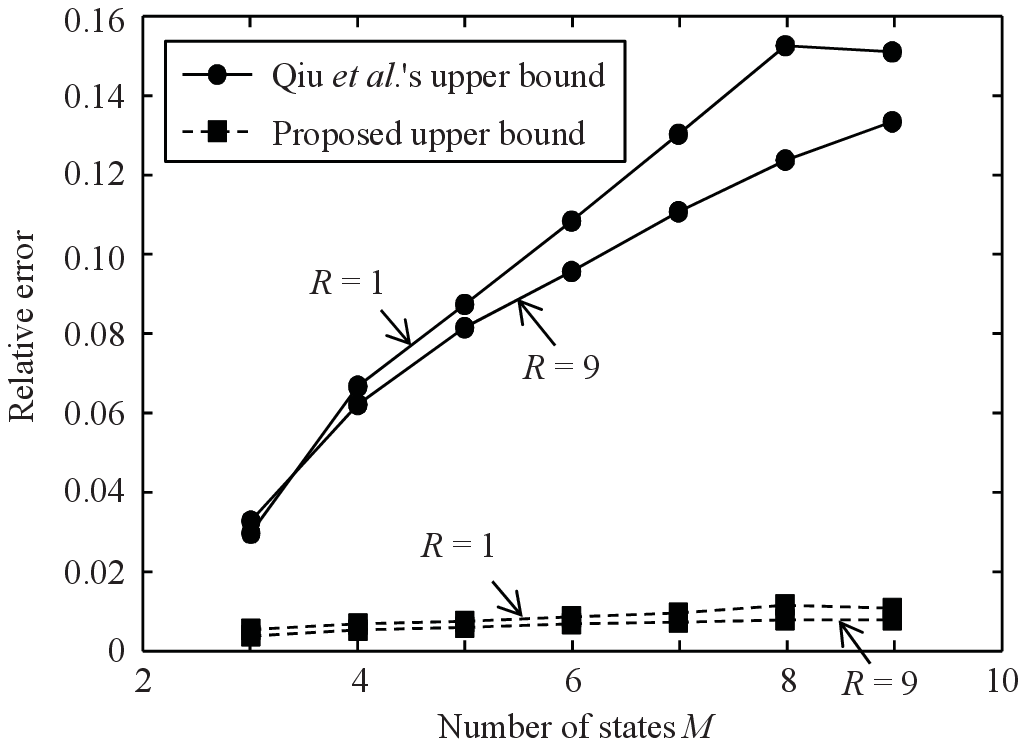}
\caption{Average relative errors of upper bounds, $\PCUp$, on the success probability
 of minimum-error measurements for $M$ quantum states.}
\label{fig:result_me_UB}
\end{figure}

Figure~\ref{fig:result_me_LB} shows the average relative errors of
the proposed lower bound, $\PCLp$, and the success probability of the SRM;
the former is more than 5.8~times smaller than the latter.

\begin{figure}[tb]
\centering
\includegraphics[scale=0.8]{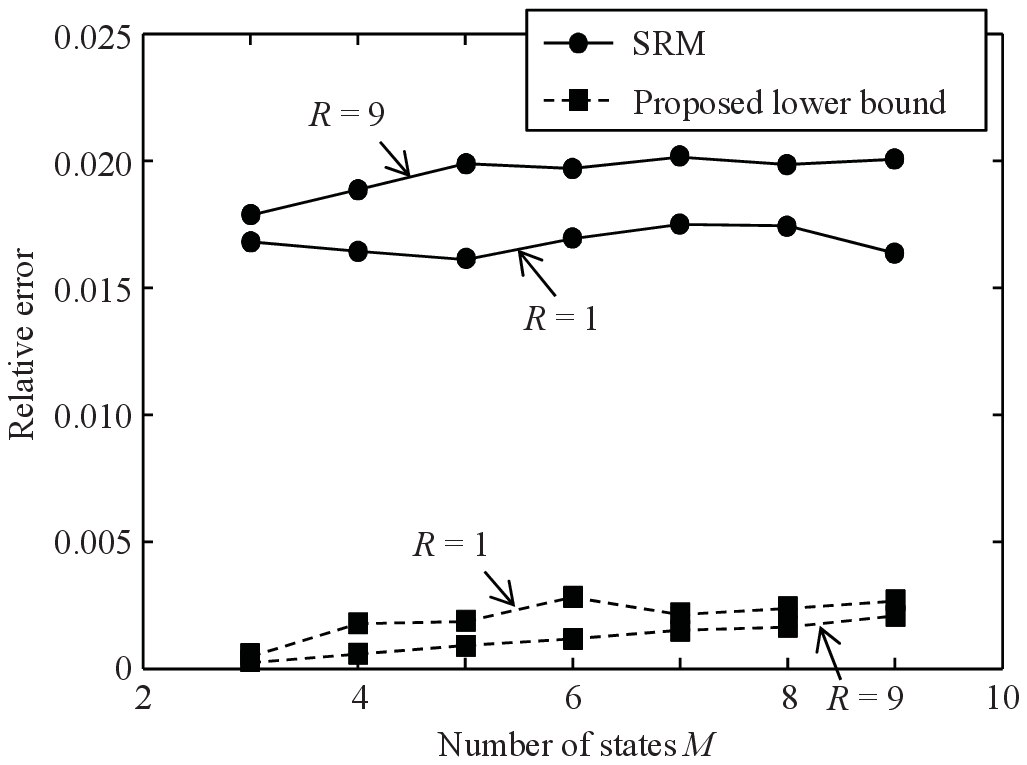}
\caption{Average relative errors of lower bounds, $\PCLp$, on the success probability
 of minimum-error measurements for $M$ quantum states.}
\label{fig:result_me_LB}
\end{figure}

\subsection{Case of optimal inconclusive measurements}

%The performance of the proposed upper bound, $\PCUip$,
%on the success probability, $\PCopt_p$, of optimal inconclusive measurements
%was evaluated as follows.
Figure~\ref{fig:result_inc_binary} shows
the average relative errors of $\PCUip$ with $J = 2$ and $3$
in the case of binary state sets.
It also shows the upper bound proposed by Sugimoto {\it et al.}
\cite{Sug-Has-Hor-Hay-2009},
which is based on the fidelity between the binary states.
In the case of $R = 1$, the analytical expression of the optimal value, $\PCopt_p$,
is given \cite{Sug-Has-Hor-Hay-2009,Nak-Usu-2012-receiver};
Sugimoto {\it et al.}'s upper bound exploits this expression,
and achieves $\PCopt_p$ when $R = 1$.
Although the proposed upper bound has a nonzero error when $R = 1$,
at least in the range of $2 \le M \le 9$,
the average relative error of $\PCUip$ is more than three times smaller than that of
Sugimoto {\it et al.}'s upper bound.

\begin{figure}[tb]
\centering
\includegraphics[scale=0.8]{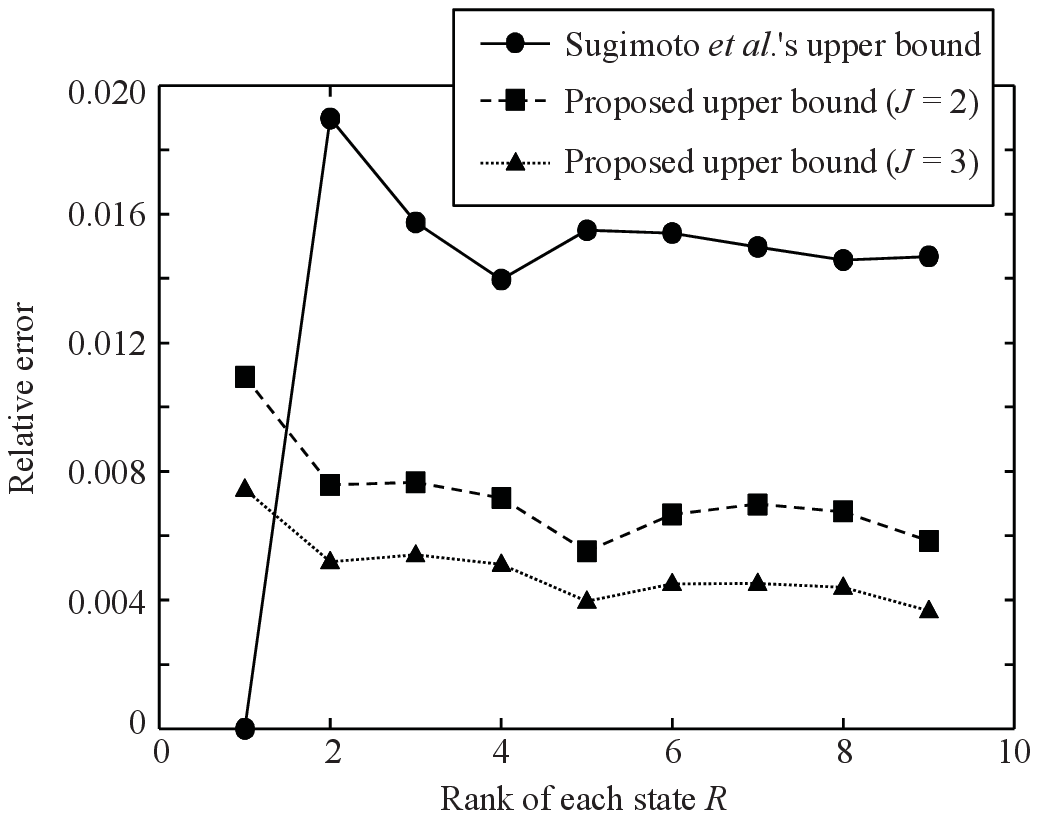}
\caption{Average relative errors of upper bounds, $\PCUip$,
 on the success probability of optimal inconclusive measurements
 for binary quantum state sets (i.e., $M = 2$).}
\label{fig:result_inc_binary}
\end{figure}

Figures~\ref{fig:result_inc_UB} and \ref{fig:result_inc_LB} respectively show
the average relative errors of the proposed upper and lower bounds,
$\PCUip$ and $\PCLip$, in the case of $M \ge 3$.
It shows that the average relative error increases gradually
with increasing $M$.
In each case, we observed that
at least in the range of $M \le 9$ and $R \le 9$
the average relative error is less than 0.037 and 0.032 with
$J = 2$ and $3$, respectively.

\begin{figure}[tb]
\centering
\includegraphics[scale=0.8]{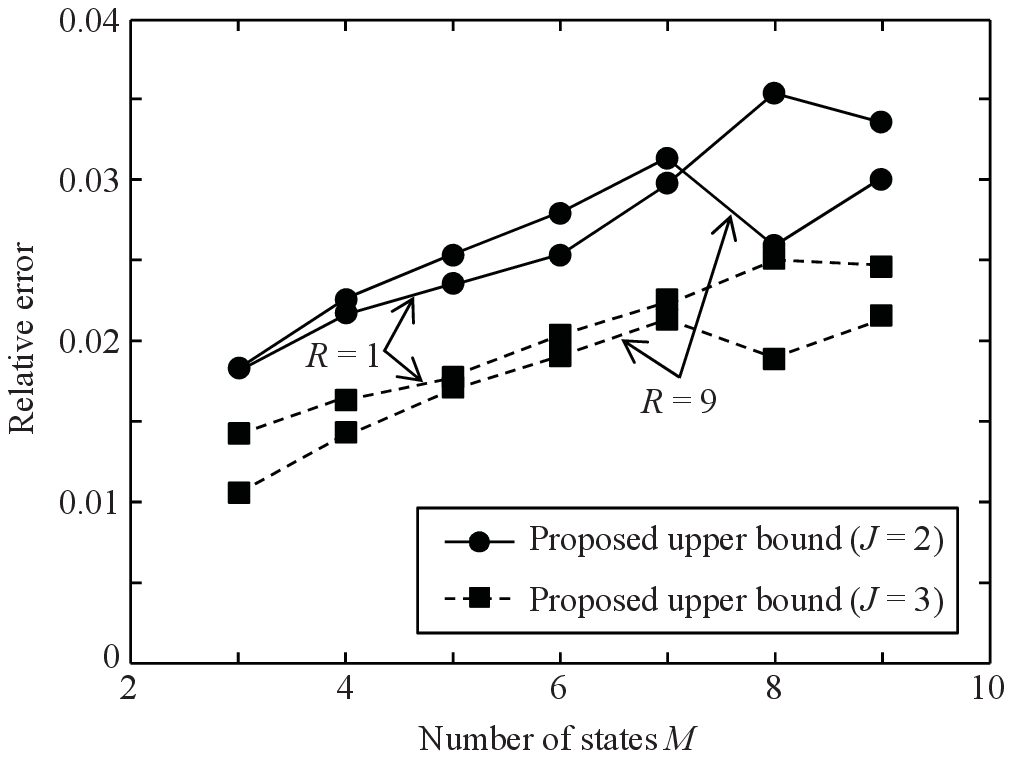}
\caption{Average relative errors of upper bounds, $\PCUip$,
 on the success probability of optimal inconclusive measurements
 for $M \ge 3$ quantum states.}
\label{fig:result_inc_UB}
\end{figure}

\begin{figure}[tb]
\centering
\includegraphics[scale=0.8]{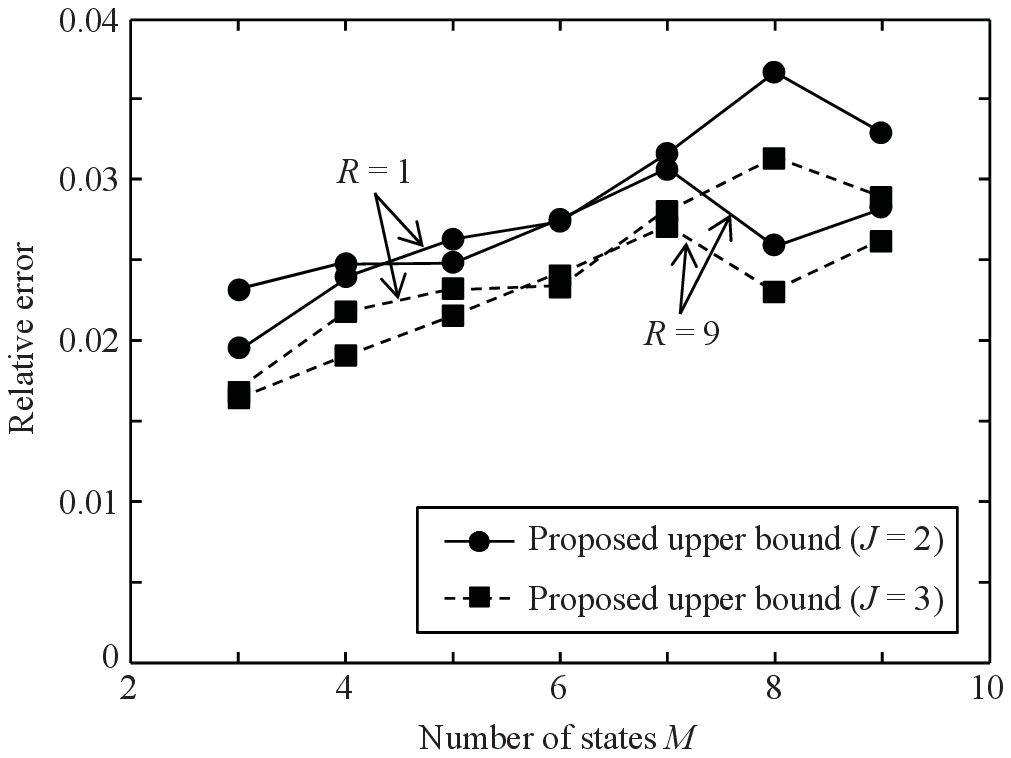}
\caption{Average relative errors of lower bounds, $\PCLip$,
 on the success probability of optimal inconclusive measurements
 for $M \ge 3$ quantum states.}
\label{fig:result_inc_LB}
\end{figure}

\section{Conclusion}

We proposed upper and lower bounds on the success probabilities of
minimum-error and optimal inconclusive measurements.
The proposed upper bounds are suboptimal solutions to the dual problems of
the optimal state discrimination problems.
The proposed lower bounds are obtained from the success probabilities
of POVMs corresponding to suboptimal solutions to the dual problems.
Numerical examples show that, on average,
the proposed upper bound for minimum-error measurements
is tighter than Qiu {\it et al.}'s one,
and the proposed bound for optimal inconclusive measurements
is tighter than Sugimoto {\it et al.}'s one in the case of
binary mixed quantum states.

\begin{acknowledgments}
 We are grateful to O. Hirota of Tamagawa University for support.
 T. S. U. was supported (in part) by JSPS KAKENHI (Grant No.16H04367).
\end{acknowledgments}

\appendix

\section{Lemmas on Hermitian operators} \label{append:Hermite}

Let $\lambda_0(\hH) \ge \lambda_1(\hH) \ge \cdots \ge \lambda_{N-1}(\hH)$ be
the ordered eigenvalues of an $N$-dimensional Hermitian operator $\hH$.

\begin{lemma} \label{lemma:AB}
 $\Tr~\hA_+ \ge \Tr~\hB_+$ holds for any Hermitian operators $\hA$ and $\hB$
 with $\hA \ge \hB$,
 where the equality holds if and only if $\hA_+ = \hB_+$.
\end{lemma}

\begin{proof}
First, we show $\Tr~\hA_+ \ge \Tr~\hB_+$.
Let $N$ be the dimension of the space on which $\hA$ and $\hB$ act.
Since $\hA \ge \hB$,
$\sum_{n=0}^k \lambda_n(\hA) \ge \sum_{n=0}^k \lambda_n(\hB)$ holds for any $k \in \mI_N$ \cite{Mar-2010}.
In contrast, for any $N$-dimensional Hermitian operator $\hH$, the following can be easily obtained:
\begin{eqnarray}
 \Tr~\hH_+ &\ge& \sum_{n=0}^k \lambda_n(\hH), ~ \forall k \in \mI_N. \label{eq:TrH}
\end{eqnarray}
Therefore, it follows that
\begin{eqnarray}
 \Tr~\hA_+ \ge \sum_{n=0}^{t-1} \lambda_n(\hA) \ge \sum_{n=0}^{t-1} \lambda_n(\hB) = \Tr~\hB_+,
\end{eqnarray}
where $t$ is the number of positive eigenvalues of $\hB$.

Next, we show that $\hA_+ = \hB_+$ holds if $\Tr~\hA_+ = \Tr~\hB_+$
(the converse is obvious).
Let $\hP = \Proja(\hB)$.
From $\hA_+ \ge \hA$, $\hA_+ \ge \hB$ holds.
Premultiplying and postmultiplying $\hA_+ \ge \hB$ with $\hP$ yields $\hP \hA_+ \hP \ge \hB_+$.
Thus, we have
\begin{eqnarray}
 \Tr~\hA_+ &\ge& \Tr(\hA_+^{1/2}\hP\hA_+^{1/2}) = \Tr(\hP\hA_+\hP) \nonumber \\
 &\ge& \Tr~\hB_+ = \Tr~\hA_+, \label{eq:lemma_AB_tr}
\end{eqnarray}
where the first inequality follows from $\hA_+ \ge \hA_+^{1/2}\hP\hA_+^{1/2}$,
which is obtained from $\ident \ge \hP$.
From Eq.~(\ref{eq:lemma_AB_tr}), $\Tr~\hA_+ = \Tr(\hP\hA_+\hP)$ holds.
It thus follows that $\supp~\hA_+ \subseteq \supp~\hP$, i.e., $\hA_+ = \hP\hA_+\hP$,
which gives $\hA_+ \ge \hB_+$.
Since $\Tr(\hA_+ - \hB_+) = 0$, $\hA_+ - \hB_+ = 0$ holds.
\QED
\end{proof}

\begin{lemma} \label{lemma:TrAB}
For any Hermitian operators $\hA$ and $\hB$, $\Tr~\hA_+ + \Tr~\hB_+ \ge \Tr(\hA + \hB)_+$ holds.
\end{lemma}

\begin{proof}
$\hA_+ \ge \hA$ and $\hB_+ \ge \hB$ gives $\hA_+ + \hB_+ \ge \hA + \hB$.
Thus, from Lemma~\ref{lemma:AB}, we have
\begin{eqnarray}
  \hspace{-1em}
   \Tr~\hA_+ + \Tr~\hB_+ = \Tr(\hA_+ + \hB_+)_+ \ge \Tr(\hA + \hB)_+.
\end{eqnarray}
\QED
\end{proof}

\begin{lemma} \label{lemma:TrABApBp}
For any Hermitian operator $\hA$ and $\hB$ with $\hA \ge \hB$,
\begin{eqnarray}
  \Tr[(\hA - \hB) \Proja(\hA)] &\ge& \Tr[(\hA - \hB) \Projb(\hB)]. \label{eq:TrABApBp}
\end{eqnarray}
\end{lemma}

\begin{proof}
Let us consider the following optimization problem:
\begin{eqnarray}
  \begin{array}{ll}
   {\rm maximize} & \Tr(\hC \hPhi) \\
   {\rm subject~to} & \ident \ge \hPhi \ge 0, \\ \label{eq:TrABApBp_main}
  \end{array}
\end{eqnarray}
where $\hC$ is a Hermitian operator.
It is clear that $\hPhi = \Proja(\hC)$ and $\hPhi = \Projb(\hC)$ are
optimal solutions to this problem.
Substituting $\hC = \hA$ and $\hC = \hB$, respectively, into problem (\ref{eq:TrABApBp_main})
gives
\begin{eqnarray}
  \Tr[\hA \Proja(\hA)] &\ge& \Tr[\hA \Projb(\hB)], \nonumber \\
  \Tr[\hB \Proja(\hA)] &\le& \Tr[\hB \Projb(\hB)].
\end{eqnarray}
Therefore, Eq.~(\ref{eq:TrABApBp}) holds.
\QED
\end{proof}

\begin{lemma} \label{lemma:AHA_ge}
For any operator $\hA$ and any positive semidefinite operators $\hrho$, $\hX$, and $\hE$
with $\ident \ge \hE \ge 0$,
\begin{eqnarray}
  \lefteqn{ \Tr[\hA[\hX + (\hrho - \hX)_+]\hA^\dag] } \nonumber \\
  &\ge& \Tr[\hA\hrho\hA^\dag(\ident - \hE) + \hA\hX\hA^\dag\hE]. \label{eq:AHA_ge_tr}
\end{eqnarray}
The equality holds if and only if
\begin{eqnarray}
  \hA(\hX - \hrho)_+\hA^\dag &=& [\hA(\hX - \hrho)\hA^\dag]_+, \label{eq:AHA_nas_notr} \\
  \Projb[\hA(\hX - \hrho)\hA^\dag] &\ge& \hE \ge \Proja[\hA(\hX - \hrho)\hA^\dag].
   \label{eq:AHA_nas_proj}
\end{eqnarray}
\end{lemma}

\begin{proof}
It follows that
\begin{eqnarray}
  \lefteqn{ \Tr[\hA[\hX + (\hrho - \hX)_+]\hA^\dag] } \nonumber \\
  &=& \Tr[\hA[\hrho + (\hX - \hrho)_+]\hA^\dag] \nonumber \\
  &\ge& \Tr[\hA\hrho\hA^\dag + (\hA\hX\hA^\dag - \hA\hrho\hA^\dag)_+] \nonumber \\
  &\ge& \Tr[\hA\hrho\hA^\dag(\ident - \hE) + \hA\hX\hA^\dag\hE], \label{eq:AHA_ge}
\end{eqnarray}
where the second line follows from $\hX + (\hrho - \hX)_+ = \hrho + (\hX - \hrho)_+$.
The third line follows from Lemma~\ref{lemma:AB} by substituting
$\hA(\hX-\hrho)_+\hA^\dag$ and $\hA(\hX-\hrho)\hA^\dag$ for $\hA$ and $\hB$, respectively.
Note that $\hA(\hX-\hrho)_+\hA^\dag \ge \hA(\hX-\hrho)\hA^\dag$ holds
from $(\hX-\hrho)_+ \ge \hX-\hrho$.
The fourth line follows from Lemma~\ref{lemma:AB_dual}.
From Lemmas~\ref{lemma:AB_dual} and \ref{lemma:AB},
the equality in Eq.~(\ref{eq:AHA_ge}) holds if and only if
Eqs.~(\ref{eq:AHA_nas_notr}) and (\ref{eq:AHA_nas_proj}) hold.
\QED
\end{proof}

\begin{lemma} \label{lemma:supp}
For any positive semidefinite operators $\hA$ and $\hB$ with $\supp~\hA \subseteq \supp~\hB$
and any operator $\hC$,
$\supp(\hC\hA\hC^\dag) \subseteq \supp(\hC\hB\hC^\dag)$ holds.
Moreover, if $\supp~\hA = \supp~\hB$, then $\supp(\hC\hA\hC^\dag) = \supp(\hC\hB\hC^\dag)$ holds.
\end{lemma}

\begin{proof}
$\supp~\hA \subseteq \supp~\hB$ gives $\Ker~\hA \supseteq \Ker~\hB$.
We obtain
\begin{eqnarray}
  \ket{x} \in \Ker(\hC\hB\hC^\dag) &\Longrightarrow& \hB^\half\hC^\dag\ket{x} = 0 \nonumber \\
  &\Longrightarrow& \hC^\dag\ket{x} \in \Ker~\hB \nonumber \\
  &\Longrightarrow& \hC^\dag\ket{x} \in \Ker~\hA \nonumber \\
  &\Longrightarrow& \hA^\half\hC^\dag\ket{x} = 0 \nonumber \\
  &\Longrightarrow& \ket{x} \in \Ker(\hC\hA\hC^\dag),
\end{eqnarray}
which indicates $\Ker(\hC\hA\hC^\dag) \supseteq \Ker(\hC\hB\hC^\dag)$,
i.e., $\supp(\hC\hA\hC^\dag) \subseteq \supp(\hC\hB\hC^\dag)$.
If $\supp~\hA = \supp~\hB$, then, from $\supp~\hA \subseteq \supp~\hB$ and $\supp~\hA \supseteq \supp~\hB$,
$\supp(\hC\hA\hC^\dag) = \supp(\hC\hB\hC^\dag)$ obviously holds.
\QED
\end{proof}

\section{Supplement of Theorem~\ref{thm:nas}} \label{append:Pi_E}

Let $\hP_{\hX}$ be the projection operator onto the support space
of a positive semidefinite operator $\hX$; i.e., $\hP_{\hX} = \Proja(\hX)$.

For any POVM $\Pi = \{ \hPi_m \}_{m \in \mI_M}$, define $\hE_m$ as
\begin{eqnarray}
 \hE_m &=& (\hA_m^\inv)^\dag \left( \sum_{k=0}^{m-1} \hPi_k \right) \hA_m^\inv, \nonumber \\
 & & ~~~ m \in \{ 1, 2, \cdots, M-1 \}, \label{eq:nas_E}
\end{eqnarray}
where $\hA_m$ is defined as Eq.~(\ref{eq:nas_Am}),
and $\hA^\inv$ denotes the Moore-Penrose inverse operator of $\hA$.
Now, we show that Eq.~(\ref{eq:Pi_A}) and $\ident \ge \hE_m \ge 0$ hold.

First, we show Eq.~(\ref{eq:Pi_A}).
From $\hA_m^\inv\hA_m = \hP_{|\hA_m|}$,
we have that for any $m \in \mI_{M-1}$,
\begin{eqnarray}
 |\hA_m|^2 &=& \hA_{m+1}^\dag\hE_{m+1}\hA_{m+1}
  = \hP_{|\hA_{m+1}|} \left( \sum_{k=0}^m \hPi_k \right) \hP_{|\hA_{m+1}|},
  \nonumber \\ \label{eq:Pi_E_A2_0}
\end{eqnarray}
where the first equality follows from Eq.~(\ref{eq:nas_Am}).
Using Eq.~(\ref{eq:Pi_E_A2_0}), we can show
\begin{eqnarray}
 |\hA_m|^2 &=& \sum_{k=0}^m \hPi_k, ~ \forall m \in \mI_M \label{eq:Pi_E_A2}
\end{eqnarray}
by induction as follows.
The case of $m = M-1$ is obvious.
Assume that Eq.~(\ref{eq:Pi_E_A2}) holds when $m = t + 1$ with $t \in \mI_{M-1}$;
we have
\begin{eqnarray}
 \supp~\hP_{|\hA_{t+1}|} &=& \supp~|\hA_{t+1}|^2 = \supp \left(\sum_{k=0}^{t+1} \hPi_k \right) \nonumber \\
 &\supseteq& \supp \left(\sum_{k=0}^t \hPi_k \right),
\end{eqnarray}
which yields $\hP_{|\hA_{t+1}|} \left( \sum_{k=0}^t \hPi_k \right) \hP_{|\hA_{t+1}|} = \sum_{k=0}^t \hPi_k$.
Thus, Eq.~(\ref{eq:Pi_E_A2}) also holds when $m = t$.
% Considering the cases of $m = M-2, M-3, \cdots, 0$ gives
% $\supp(\sum_{k=0}^m \hPi_k) \subseteq \supp~\hP_{|\hA_{m+1}|}$
% for any $m \in \mI_{M-1}$.
% It thus follows that
% %
% \begin{eqnarray}
%  |\hA_m|^2 &=& \sum_{k=0}^m \hPi_k. \label{eq:Pi_E_A2}
% \end{eqnarray}
% %
% Obviously, Eq.~(\ref{eq:Pi_E_A2}) also holds when $m = M-1$.
Equation~(\ref{eq:Pi_A}) is readily obtained from Eq.~(\ref{eq:Pi_E_A2}).

Next, we show $\ident \ge \hE_m \ge 0$.
$\hE_m \ge 0$ obviously holds, so we only need to show $\ident \ge \hE_m$.
From Eq.~(\ref{eq:Pi_A}), $|\hA_{m-1}|^2 \le |\hA_m|^2$ holds.
Premultiplying and postmultiplying $|\hA_{m-1}|^2 \le |\hA_m|^2$
with $(\hA_m^\inv)^\dag$ and $\hA_m^\inv$, respectively,
and using $\hA_m\hA_m^\inv = \hP_{\hA\hA^\dag}$,
we have that for any $m$ with $1 \le m \le M-1$,
\begin{eqnarray}
 \hE_m &=& (\hA_m^\inv)^\dag |\hA_{m-1}|^2 \hA_m^\inv \le (\hA_m^\inv)^\dag |\hA_m|^2 \hA_m^\inv \nonumber \\
 &=& (\hA_m\hA_m^\inv)^\dag (\hA_m\hA_m^\inv) = \hP_{\hA\hA^\dag} \le \ident.
\end{eqnarray}

\section{Proof of Corollary~\ref{cor:nas_supp}} \label{append:cor_nas_supp}

The necessity is obvious from Theorem~\ref{thm:nas}.
We prove the sufficiency as follows.
Assume $\PCUp = \PCopt$.
We choose $\{ \hE_k \}_{k=1}^{M-1}$ satisfying Eqs.~(\ref{eq:nas_Pi_cond}) and (\ref{eq:nas_commute})
(such $\{ \hE_k \}_{k=1}^{M-1}$ exists from Theorem~\ref{thm:nas}).
To show Eq.~(\ref{eq:nas_commute2}),
it is sufficient to show that the following equations hold for any $m$ with $1 \le m \le M-1$:
\begin{eqnarray}
 \hA_m (\hX_{m-1} - \hrho_m)_+ \hA_m^\dag &=& \ha_m (\hX_{m-1} - \hrho_m)_+ \ha_m^\dag,
  \label{eq:nas_supp_Aa1} \\
 \hA_m (\hX_{m-1} - \hrho_m) \hA_m^\dag &=& \ha_m (\hX_{m-1} - \hrho_m) \ha_m^\dag,
  \label{eq:nas_supp_Aa2}
\end{eqnarray}
where $\hA_m$ is defined by Eq.~(\ref{eq:nas_Am}).
$(\hX_{m-1} - \hrho_m)_+$, $\hX_{m-1}$, and $\hrho_m$ are positive semidefinite operators
whose support spaces are subspaces of $\supp~\hX_m$.
Thus, if
\begin{eqnarray}
 \hA_m\hx\hA_m^\dag &=& \ha_m\hx\ha_m^\dag,
  ~~~ \forall \hx \ge 0 ~{\rm s.t.}~\supp~\hx \subseteq \supp~\hX_m \nonumber \\
 \label{eq:nas_supp_AxA}
\end{eqnarray}
for any $m$ with $1 \le m \le M-1$,
then substituting $(\hX_{m-1} - \hrho_m)_+$, $\hX_{m-1}$, and $\hrho_m$ into
$x$ in Eq.~(\ref{eq:nas_supp_AxA}) gives
Eqs.~(\ref{eq:nas_supp_Aa1}) and (\ref{eq:nas_supp_Aa2}).
Therefore, it suffices to show that Eq.~(\ref{eq:nas_supp_AxA}) holds for any $m$ with $1 \le m \le M-1$.

In preparation for proving it,
we show that if Eq.~(\ref{eq:nas_supp_AxA}) holds for a certain $m$ with $1 \le m \le M-1$,
then we have that for any Hermitian operator $\hy$ with $\supp~\hy \subseteq \mR_m$,
\begin{eqnarray}
 \hE_m^\half\hy\hE_m^\half &=& \he_m\hy\he_m, \label{eq:nas_supp_EyE}
\end{eqnarray}
where $\mR_m = \supp(\ha_m\hX_m\ha_m^\dag)$.
Let $\hT_m = \ha_m(\hX_{m-1} - \hrho_m)\ha_m^\dag$; then,
from Eq.~(\ref{eq:suppX}), we have $\mR_m = \supp~\hT_m$.
Let $\ul{\he_m} = \Proja(\hT_m)$.
Recall $\he_m = \Projb(\hT_m)$.
For any $\he$ with $\he_m \ge \he \ge \ul{\he_m}$,
$\supp~\Delta\he \subseteq \Ker~\hT_m$ holds,
where $\Delta\he = \he - \ul{\he_m}$,
which indicates $\supp~\Delta\he$ is perpendicular to $\mR_m$.
Thus, for any Hermitian operator $\hy$ with $\supp~\hy \subseteq \mR_m$,
from $\Delta\he \hy = \hy \Delta\he = 0$, we obtain
\begin{eqnarray}
 \he\hy\he &=& (\ul{\he_m} + \Delta\he)\hy(\ul{\he_m} + \Delta\he) = \ul{\he_m}\hy\ul{\he_m}.
  \label{eq:nas_supp_EyE2}
\end{eqnarray}
In contrast, $m$ satisfying Eq.~(\ref{eq:nas_supp_AxA}) also satisfies Eq.~(\ref{eq:nas_supp_Aa2}),
which yields $\he_m = \Projb[\hA_m (\hX_{m-1} - \hrho_m) \hA_m^\dag]$
and $\ul{\he_m} = \Proja[\hA_m (\hX_{m-1} - \hrho_m) \hA_m^\dag]$.
Accordingly, from Eq.~(\ref{eq:nas_Pi_cond}), $\he_m \ge \hE_m \ge \ul{\he_m}$ holds;
thus, $\he_m^\half \ge \hE_m^\half \ge \ul{\he_m^\half}$ holds.
From $\he_m^\half = \he_m$ and $\ul{\he_m^\half} = \ul{\he_m}$,
this gives $\he_m \ge \hE_m^\half \ge \ul{\he_m}$.
Therefore, substituting $\he = \he_m$ and $\he = \hE_m^\half$ into Eq.~(\ref{eq:nas_supp_EyE2})
gives
\begin{eqnarray}
 \hE_m^\half\hy\hE_m^\half &=& \ul{\he_m}\hy\ul{\he_m} = \he_m\hy\he_m,
\end{eqnarray}
i.e., Eq.~(\ref{eq:nas_supp_EyE}) holds.

We prove Eq.~(\ref{eq:nas_supp_AxA}) for any $m$ with $1 \le m \le M-1$ by induction on $m$.
This is obvious for $m = M-1$, since $\hA_{M-1} = \ha_{M-1} = \ident$ holds.
Assume that, for a certain $m = k+1 \le M-1$, Eq.~(\ref{eq:nas_supp_AxA}) holds.
For any $\hx \ge 0$ with $\supp~\hx \subseteq \supp~\hX_k$, we obtain
\begin{eqnarray}
 \hA_k\hx\hA_k^\dag &=& \hE_{k+1}^\half\hA_{k+1}\hx\hA_{k+1}^\dag\hE_{k+1}^\half \nonumber \\
 &=& \hE_{k+1}^\half\ha_{k+1}\hx\ha_{k+1}^\dag\hE_{k+1}^\half \nonumber \\
 &=& \he_{k+1}\ha_{k+1}\hx\ha_{k+1}^\dag\he_{k+1} \nonumber \\
 &=& \ha_k\hx\ha_k^\dag,
\end{eqnarray}
where the second line follows from $\supp~\hx \subseteq \supp~\hX_k \subseteq \supp~\hX_{k+1}$
and Eq.~(\ref{eq:nas_supp_AxA}) with $m = k+1$.
The third line follows from
$\supp(\ha_{k+1}\hx\ha_{k+1}^\dag) \subseteq \supp(\ha_{k+1}\hX_{k+1}\ha_{k+1}^\dag) = \mR_{k+1}$,
which is obtained by $\supp~\hx \subseteq \supp~\hX_{k+1}$ and Lemma~\ref{lemma:supp},
and from Eq.~(\ref{eq:nas_supp_EyE}) with $m = k+1$ and $\hy = \ha_{k+1}\hx\ha_{k+1}^\dag$.
Therefore, Eq.~(\ref{eq:nas_supp_AxA}) holds for $m = k$.
\QED

\section{Proof of Proposition~\ref{pro:g}} \label{append:g}

(1) We have
\begin{eqnarray}
 \lambda_{\min}(\hG^{-1/2} \hX^\sopt \hG^{-1/2}) \ge a
  &\iff& \hG^{-1/2} \hX^\sopt \hG^{-1/2} \ge a \ident \nonumber \\
 &\iff& \hX^\sopt \ge a \hG. \label{eq:XaG}
\end{eqnarray}
From Eq.~(\ref{eq:sa}), $s(a) = \Tr~\hX^\sopt - a p$ holds when $\hX^\sopt \ge a \hG$.
Moreover, $\hX^\sopt \ge \hrho_m$ for any $m \in \mI_M$ gives
\begin{eqnarray}
 \hX^\sopt - \frac{\hG}{M} &=& \frac{1}{M} \sum_{m=0}^{M-1}(\hX^\sopt - \hrho_m) \ge 0.
\end{eqnarray}
Thus, from Eq.~(\ref{eq:XaG}) with $a = 1/M$,
$1/M \le \lambda_{\min}(\hG^{-1/2} \hX^\sopt \hG^{-1/2})$.

(2) We have
\begin{eqnarray}
 a \ge \lambda_{\max}(\hG^{-1/2} \hX^\sopt \hG^{-1/2})
  &\iff& a \ident \ge \hG^{-1/2} \hX^\sopt \hG^{-1/2} \nonumber \\
 &\iff& a \hG \ge \hX^\sopt. \label{eq:XaG2}
\end{eqnarray}
Thus, $a\hG - \hX^\sopt \ge 0$.
From Eq.~(\ref{eq:sa}) and $\Tr~\hG = 1$, we have
\begin{eqnarray}
 s(a) &=& \Tr~\hX^\sopt + \Tr(a \hG - \hX^\sopt) - a p = a (1 - p).
\end{eqnarray}

(3) For any $t$ with $0 \le t \le 1$ and $a, a' \in \Real_+$,
substituting $\hA = t (a\hG - \hX^\sopt)$ and $\hB = (1-t) (a'\hG - \hX^\sopt)$ into
Lemma~\ref{lemma:TrAB} gives
\begin{eqnarray}
 \lefteqn{ t \Tr(a\hG - \hX^\sopt)_+ + (1-t) \Tr(a'\hG - \hX^\sopt)_+ } \nonumber \\
 &\ge& \Tr[[ta + (1-t)a']\hG - \hX^\sopt]_+,
\end{eqnarray}
where we use $\hA + \hB = [ta + (1-t)a']\hG - \hX^\sopt$.
Therefore, from Eq.~(\ref{eq:sa}), $ts(a) + (1-t)s(a') \ge s[ta + (1-t)a']$ obviously holds;
i.e., $s(a)$ is convex.
\QED

\section{Proof of Proposition~\ref{pro:sa}} \label{append:sa}

Let $\hPhi_a = \Proja(a\hG - \hX^\sopt)$ and $\hPhi_a^+ = \Projb(a\hG - \hX^\sopt)$;
then, $\tp(a) = \Tr(\hG\hPhi_a)$ and $\tp^+(a) = \Tr(\hG\hPhi_a^+)$ hold.

(1) For any $a, a' \in \Real_+$ with $a < a'$, we have
\begin{eqnarray}
 (a'-a) \Tr(\hG\hPhi_{a'}) &\ge& (a'-a) \Tr(\hG\hPhi_a^+), \label{eq:sa_TrGPhi}
\end{eqnarray}
which follows from substituting $\hA = a'\hG - \hX^\sopt$ and $\hB = a\hG - \hX^\sopt$ into
Lemma~\ref{lemma:TrABApBp}.
Dividing both sides of Eq.~(\ref{eq:sa_TrGPhi}) by $a'-a$ yields $\tp(a') \ge \tp^+(a)$.
In contrast, since $\tp(b) \le \tp^+(b)$ for any $b \in \Real_+$,
we obtain
\begin{eqnarray}
 \tp(a) &\le& \tp^+(a) \le \tp(a') \le \tp^+(a'),
\end{eqnarray}
which indicates that $\tp(a)$ and $\tp^+(a)$ monotonically increase with respect to $a$.

(2) First, we show $\tp(a^\opt) \le p \le \tp^+(a^\opt)$,
where $a^\opt \in \argmin_a s(a)$.
The dual problem of problem (\ref{eq:inc_dual2t}) is expressed as
(see Ref.~\cite{Nak-Kat-Usu-2015-general}):
\begin{eqnarray}
 \begin{array}{ll}
  {\rm maximize} & \Tr[\hX^\sopt(\ident - \hPhi)] \\
  {\rm subject~to} & \ident \ge \hPhi \ge 0, \Tr(\hG\hPhi) = p. \\
 \end{array} \label{eq:inc_main2t}
\end{eqnarray}
Let $\hPhi^\opt$ be an optimal solution to problem (\ref{eq:inc_main2t}).
Since the optimal value of problem (\ref{eq:inc_dual2t}), $s(a^\opt)$,
is equivalent to the optimal value of problem (\ref{eq:inc_main2t}),
we have
\begin{eqnarray}
 s(a^\opt) &=& \Tr[\hX^\sopt(\ident - \hPhi^\opt)] \nonumber \\
 &=& \Tr[\hX^\sopt(\ident - \hPhi^\opt)] + \Tr(a\hG\hPhi^\opt) - ap, \label{eq:saopt}
\end{eqnarray}
where the second line follows from $\Tr(\hG\hPhi^\opt) = p$.
In contrast, $s(a) + ap$ is equivalent to the optimal value of
problem (\ref{eq:AB_dual}) with $\hA = a\hG$ and $\hB = \hX^\sopt$.
Thus, from Eq.~(\ref{eq:AB_main}),
we have that for any operator $\hPhi$ with $\ident \ge \hPhi \ge 0$,
\begin{eqnarray}
 s(a) + ap &\ge& \Tr[\hX^\sopt(\ident - \hPhi)] + \Tr(a\hG\hPhi). \label{eq:saopt2}
\end{eqnarray}
From Lemma~\ref{lemma:AB_dual},
if the equality in Eq.~(\ref{eq:saopt2}) holds, then $\hPhi_a^+ \ge \hPhi \ge \hPhi_a$ holds.
Thus, Eq.~(\ref{eq:saopt}) gives
$\hPhi_{a^\opt}^+ \ge \hPhi^\opt \ge \hPhi_{a^\opt}$.
Multiplying both sides by $\hG$ and taking the trace gives
$\tp(a^\opt) \le p \le \tp^+(a^\opt)$.

Next assume that $\tp(a) \le p \le \tp^+(a)$; we show that $a$ minimizes $s(a)$.
Since problem (\ref{eq:inc_main2t}) is the dual problem of problem (\ref{eq:inc_dual2t}),
we have that for any operator $\hPhi$ with $\ident \ge \hPhi \ge 0$
and $\Tr(\hG\hPhi) = p$,
\begin{eqnarray}
 s(a) &\ge& s(a^\opt) \ge \Tr[\hX^\sopt(\ident - \hPhi)].
\end{eqnarray}
Thus, to prove that $a$ minimizes $s(a)$, i.e., $s(a) = s(a^\opt)$,
it suffices to find $\hPhi$ with
$\ident \ge \hPhi \ge 0$ and $\Tr(\hG\hPhi) = p$ such that $s(a) = \Tr[\hX^\sopt(\ident - \hPhi)]$.
We show that $\hPhi = c \hPhi_a + (1-c) \hPhi_a^+$ is such a value,
where $c = 1$ if $\tp(a) = \tp^+(a)$; otherwise, $c = [\tp^+(a) - p] / [\tp^+(a) - \tp(a)]$.
Note that $c$ obviously satisfies $0 \le c \le 1$.
It is easily seen that $\ident \ge \hPhi \ge 0$ and $\Tr(\hG\hPhi) = p$ hold.
From $\hPhi_a^+ \ge \hPhi_a$, $\hPhi_a^+ \ge \hPhi \ge \hPhi_a$ holds.
Substituting $\hA = a\hG$ and $\hB = \hX^\sopt$ into Lemma~\ref{lemma:AB_dual}
and using Eq.~(\ref{eq:AB_main}) gives
\begin{eqnarray}
 \hspace{-1em}
  \Tr~\hX^\sopt + \Tr(a\hG - \hX^\sopt)_+ &=& \Tr[\hX^\sopt(\ident - \hPhi)] + \Tr(a\hG\hPhi) \nonumber \\
 &=& \Tr[\hX^\sopt(\ident - \hPhi)] + ap,
\end{eqnarray}
where the second line follows from $\Tr(\hG\hPhi) = p$.
Therefore, from Eq.~(\ref{eq:sa}), $s(a) = \Tr[\hX^\sopt(\ident - \hPhi)]$ holds.
\QED

%\bibliography{quant}

\end{document}